\documentclass[a4paper,UKenglish,cleveref, autoref,numberwithinsect]{lipics-v2019}
\newif\iflong\longtrue
\nolinenumbers
\hideLIPIcs

\usepackage[utf8]{inputenc}
\usepackage[T1]{fontenc}
\usepackage{amsmath,amssymb}
\usepackage[arrow, matrix, curve]{xy}
\usepackage{amsthm}
\usepackage[disable]{todonotes}
\usepackage{graphicx}
\usepackage{hyperref}
\usepackage{float}
\usepackage{dsfont}
\usepackage{thm-restate}
\usepackage{braket}
\usepackage{tikz}
\usepackage{verbatim}
\usepackage{enumitem}
\usepackage{mathrsfs}
\usepackage{algorithm}

\usepackage[noend]{algpseudocode}

\definecolor{myred}{rgb}{1,0.25,0.25}

\newtheorem{redrule}{Rule}[section]
\theoremstyle{definition}

\newcommand{\Oh}{\ensuremath{\mathcal{O}}}

\theoremstyle{definition}

\newcommand{\todog}[1]{\todo[color=yellow!90]{ #1}}

\definecolor{myred}{rgb}{1,0.25,0.25}

\newcommand{\pe}{\xi}

\newcommand{\pc}{\lambda_{c}}
\newcommand{\core}{\mathscr{C}}
\newcommand{\peri}{\mathscr{P}}
\newcommand{\out}{\text{out}}

\renewcommand{\mod}{\text{ mod }}
\newcommand{\aux}{\text{aux}}
\newcommand{\cc}{component cover}
\newcommand{\badstuffhappens}{\text{NP} \subseteq \text{coNP} / \text{poly}}

\title{Maximum Edge-Colorable Subgraph and Strong Triadic Closure Parameterized by Distance to Low-Degree Graphs}
\titlerunning{Maximum Edge-Colorable Subgraph and Strong Triadic Closure}

\author{Niels Grüttemeier}{Philipps-Universität Marburg, Fachbereich Mathematik und Informatik,  Marburg, Germany}{niegru@informatik.uni-marburg.de}{https://orcid.org/0000-0002-6789-2918}{}

\author{Christian Komusiewicz}{Philipps-Universität Marburg, Fachbereich Mathematik und Informatik,  Marburg, Germany}{komusiewicz@informatik.uni-marburg.de}{https://orcid.org/0000-0003-0829-7032}{}

\author{Nils Morawietz}{Philipps-Universität Marburg, Fachbereich Mathematik und Informatik,  Marburg, Germany }{morawietz@informatik.uni-marburg.de}{}{}

\authorrunning{N.~Grüttemeier, C.~Komusiewicz, N.~Morawietz}

\Copyright{Niels~Grüttemeier, Christian~Komusiewicz, Nils~Morawietz}



\begin{document}

\maketitle
\begin{abstract}
  Given an undirected graph~$G$ and integers~$c$ and~$k$, the \textsc{Maximum Edge-Colorable
    Subgraph} problem asks whether we can delete at most~$k$ edges in~$G$ to obtain a graph
  that has a proper edge coloring with at most~$c$ colors. We show that \textsc{Maximum
    Edge-Colorable Subgraph} admits, for every fixed~$c$, a linear-size problem kernel when
  parameterized by the edge deletion distance of~$G$ to a graph with maximum degree~$c-1$. This
  parameterization measures the distance to instances that, due to Vizing's famous theorem, are
  trivial yes-instances. For~$c\le 4$, we also provide a linear-size kernel for the same
  parameterization for \textsc{Multi Strong Triadic Closure}, a related edge coloring problem
  with applications in social network analysis.  We provide further results for \textsc{Maximum
    Edge-Colorable Subgraph} parameterized by the vertex deletion distance to graphs where
  every component has order at most~$c$ and for the list-colored versions of both problems.
\end{abstract}
\todo[inline]{ToDo:\\
- List-Problems raus (?)\\
- Parameter verteidigen\\
- $\pe_{c-1}$ and~$\pc$ are not related\\
- Granovetter zitieren\\
}
\section{Introduction}
Edge coloring and its many variants form a fundamental problem family in algorithmic graph
theory~\cite{CH82,HK86,Hol81,JT11}. In the classic \textsc{Edge Coloring}
problem, the input is a graph~$G$ and an integer~$c$ and the task is to decide whether~$G$ has a \emph{proper edge coloring}, that is, an assignment of colors to the
edges of a graph such that no pair of incident edges receives the same color, with at most~$c$ colors.  The number of necessary colors for a proper
edge coloring of a graph~$G$ is closely related to the degree of~$G$: Vizing's
famous theorem states that any graph~$G$ with maximum degree~$\Delta$ can be edge-colored
with~$\Delta+1$ colors~\cite{V64}, an early example of an additive approximation
algorithm. Later it was shown that \textsc{Edge Coloring} is NP-hard for~$c=3$~\cite{Hol81},
and in light of Vizing's result it is clear that the hard instances for~$c=3$ are exactly the
subcubic graphs. Not surprisingly, the NP-hardness extends to every fixed~$c \geq 3$~\cite{LG83}.

In the more general \textsc{Maximum Edge-Colorable Subgraph (ECS)} problem, we are given an
additional integer~$k$ and want to decide whether we can delete at most~$k$ edges in the input
graph~$G$ so that the resulting graph has a proper~edge coloring with~$c$~colors. \textsc{ECS} is NP-hard
for~$c=2$~\cite{FOW02} and it has received a considerable amount of interest for small constant values
of~$c$ such as~$c=2$~\cite{FOW02,Kos09},~$c=3$~\cite{Kos09,Kow09,Riz09},
and~$c\le 7$~\cite{KK14}. Feige et al.~\cite{FOW02} mention that \textsc{ECS} has applications in call admittance in telecommunication networks. Given the large amount of algorithmic literature on this problem, it
is surprising that there is, to the best of our knowledge, no work on fixed-parameter
algorithms for ECS. This lack of interest
may be rooted in the NP-hardness of \textsc{Edge Coloring} for every fixed~$c\ge 3$, which implies that
ECS is not fixed-parameter tractable with respect to~$k+c$ unless P=NP.

Instead of the parameter~$k$, we consider the parameter~$\pe_{c-1}$ which we define as the minimum
number of edges that need to be deleted in the input graph to obtain a graph with maximum
degree~$c-1$. This is a distance-from-triviality parameterization~\cite{GHN04}: Due to
Vizing's Theorem, the answer is always yes if the input graph has maximum degree~$c-1$. We
parameterize by the edge-deletion distance to this trivial case. Observe that the number of vertices with degree at least~$c$ is at most~$2 \pe_{c-1}$. If we consider \textsc{Edge Coloring} instead of \textsc{ECS}, the instances with maximum degree larger than~$c$ are trivial no-instances. Thus, in non-trivial instances, the parameter~$\pe_{c-1}$ is essentially the same as the number of vertices that have degree~$c$. This is, arguably, one of the most natural parameterizations for \textsc{Edge Coloring}. We achieve a  kernel that has linear size
for every fixed~$c$.
\begin{restatable}{theorem}{xiecs}
  \label{Theorem: ECS xi kernel}
\textsc{ECS} admits a problem kernel with at most~$4 \pe_{c-1} \cdot c$ vertices and~$\Oh(\xi_{c-1} \cdot c^2)$ edges that can be computed in~$\Oh(n+m)$ time.
\end{restatable}
Herein,~$n$ denotes the number of vertices of the input graph~$G$ and~$m$ denotes the
number of edges.  This kernel is obtained by making the following observation about the
proof of Vizing's Theorem: When proving that an edge can be safely colored with one of~$c$
colors, we 
only need to consider the closed neighborhood of one endpoint of this edge. This allows us to show that all vertices
which have degree at most~$c-1$ and only neighbors of degree at most~$c-1$ can be safely
removed.

Next, we consider ECS parameterized by the size~$\pc$ of a smallest vertex set~$D$ such
that deleting~$D$ from~$G$ results in a graph where each connected component has at
most~$c$ vertices. The parameter~$\pc$ presents a different distance-from-triviality
parameterization, since a graph with connected components of order at most~$c$ can
trivially be colored with~$c$ edge colors. Moreover, observe that~$\pc$ is never larger than the vertex cover number which is a popular structural parameter. Again, we obtain a linear-vertex kernel for~$\pc$
when~$c$ is fixed.
\begin{restatable}{theorem}{vckernel}
  \label{Theorem: Vertex cover ECS kernel}
  \textsc{ECS} admits a problem kernel with~$\Oh(c^3 \cdot \pc)$ vertices.
\end{restatable}

We then consider \textsc{Multi Strong Triadic Closure (Multi-STC)} a closely related
edge coloring problem with applications in social network analysis~\cite{ST14}. In \textsc{Multi-STC}, we are given a graph~$G$ and two integers~$k$ and~$c$ and aim to find a coloring of
the edges with one \emph{weak} and at most~$c$ \emph{strong colors} such that every pair of
incident edges that forms an induced path on three vertices does not receive the same strong
color and the number of weak edges is at most~$k$. The idea behind this problem is to uncover the different strong relation types in social networks by using the following assumption: if one person has for example two colleagues, then these two people know each other and should also be connected in the social network. 
 In other words, if a vertex has two neighbors that are not adjacent to each other, then this is evidence that either the strong interaction types with these two neighbors are different or one of the interaction types is merely weak. \iflong

\fi Combinatorically, there are two crucial
differences to ECS: First, two incident edges may receive the same strong color if the subgraph
induced by the endpoints is a triangle. Second, instead of deleting edges to obtain a graph
that admits such a coloring, we may label edges as weak. In~ECS this does not make a
difference; in \textsc{Multi-STC}, however, deleting an edge may destroy triangles which would add an additional
constraint on the coloring of the two remaining triangle edges.

In contrast to ECS, \textsc{Multi-STC} is NP-hard already for~$c=1$~\cite{ST14}. This special case is known
as \textsc{Strong Triadic Closure} (STC). Not surprisingly, \textsc{Multi-STC} is NP-hard
for all fixed~$c\ge 2$~\cite{BGKS19}. Moreover, for~$c\ge 3$ \textsc{Multi-STC} is NP-hard even if~$k=0$, that is, even if
every edge has to be colored with a strong color. STC and Multi-STC
have received a considerable amount of interest
recently~\cite{ST14,GHK+18,GK18,BGKS19,KNP17,KP17}.

Since the edge coloring for \textsc{Multi-STC} is a relaxed version of a proper edge coloring,
we may observe that Vizing's Theorem implies the following: If the input graph~$G$ has degree at
most~$c-1$, then the instance is a yes-instance even for~$k=0$. Hence, it is very natural to
apply the parameterization by~$\pe_{c-1}$ also for \textsc{Multi-STC}. We succeed to transfer
the kernelization result from ECS to \textsc{Multi-STC} for~$c\le 4$.  In fact, our result
for~$c=3$ and~$c=4$ can be extended to the following more general result.

\begin{theorem} \label{Theorem: STC xi} \textsc{Multi-STC} admits a problem kernel
  with~$\Oh(\pe_{\lfloor \frac{c}{2} \rfloor + 1} \cdot c)$ vertices and
  $\Oh(\pe_{\lfloor \frac{c}{2} \rfloor + 1} \cdot c^2)$~edges, when limited to instances with~$c \geq 3$. The kernel can be computed
  in~$\mathcal{O}(n+m)$ time.
\end{theorem}
For~$c=5$, this gives a linear-size kernel for the parameter~$\pe_3$, for~$c=6$, a
linear-size kernel for the parameter~$\pe_4$ and so on.  Our techniques to prove
Theorem~\ref{Theorem: STC xi} are very loosely inspired by the proof of Vizing's Theorem
but in the context of \textsc{Multi-STC} several obstacles need to be overcome. As a
result, the proof differs quite substantially from the one for ECS. Moreover, in contrast do \textsc{ECS}, \textsc{Multi-STC} does not admit a polynomial kernel when parameterized by the vertex cover number~\cite{GK18} which excludes almost all popular structural parameters.

We then show how far our kernelization for~$\pe_{t}$ can be lifted to generalizations of ECS
and \textsc{Multi-STC} where each edge may choose its color only from a specified list of
colors, denoted as \textsc{Edge List ECS (EL-ECS)} and \textsc{Edge List Multi-STC (EL-Multi-STC)}. We show that for~$\pe_2$ we
obtain a linear kernel for every fixed~$c$. 
\begin{restatable}{theorem}{elkernel}
  \label{kernel}
For all~$c\in \mathds{N}$, \textsc{EL-ECS} and \textsc{EL-Multi-STC} admit an~$11\pe_2$-edge and~$10\pe_2$-vertex kernel for \textsc{EL-ECS} that can be computed in~$\Oh(n^2)$ time.
\end{restatable}
For~$c=3$, this extends Theorem~\ref{Theorem: ECS xi kernel} to the list colored version of~\textsc{ECS}. For~$c>3$ parameterization by~$\pe_2$ may seem a bit uninteresting compared to the results for \textsc{ECS} and
\textsc{Multi-STC}. However,
Theorem~\ref{kernel} is unlikely to be improved by considering~$\pe_t$ for~$t>2$.
\begin{restatable}{proposition}{elhard}\label{prop:el-hard}
  \textsc{EL-ECS} and \textsc{EL-Multi-STC} are NP-hard for all~$c \geq 3$ on~triangle-free cubic graphs even if~$\pe_3=k=0$.
\end{restatable}
\iflong
A summary of our results is shown in Table~\ref{tab:pc-results}. 
\begin{table}
\caption{A summary of our results for the two problems. Herein,~$\pe_{c-1}$ denotes the edge-deletion distance to graphs with maximum degree at most~$c-1$, and $\lambda_{c}$ denotes the vertex-deletion distance to graphs where every connected component has order at most~$c$.}
  \begin{tabular}{l c c}
\hline

Parameter & \textsc{ECS} & \textsc{Multi-STC}\\
\hline
$(\pe_{c-1},c)$ & $\mathcal{O}(\pe_{c-1}c)$-vertex kernel (Thm. \ref{Theorem: ECS xi kernel})
& $\mathcal{O}(\pe_{c-1})$-edge kernel (Thm. \ref{Theorem: STC xi kernel}),\\
&  &if~$c \le 4$ \\
\hline
$(\lambda_c,c)$& $\mathcal{O}(c^3 \cdot \lambda_c)$-vertex kernel (Thm.~\ref{Theorem: Vertex cover ECS kernel})&
~~~No poly Kernel, even for~$c=1$ \cite{GK18}~~~\\
\hline
\end{tabular}
\label{tab:pc-results}
\end{table}
\subparagraph{Organization of the Paper.} In Section~\ref{sec:prelim}, we set
the notation for this work and present the formal definition of all problems under
consideration. In Section~\ref{sec:ed}, we show the kernels for ECS parameterized by the
edge-deletion distance to low-degree graphs and the vertex-deletion distance to graphs
with small connected components. In Section~\ref{sec:stc}, we show the kernels for
Multi-STC parameterized by the edge-deletion distance to low-degree graphs. Finally, in
Section~\ref{sec:lists}, we consider the edge deletion distance to degree-two graphs for
the problem variants with edge lists.  \fi 

\iflong \section{Preliminaries}\label{sec:prelim}\fi
\subparagraph{Notation.} We consider simple undirected graphs~$G=(V,E)$ \iflong where~$n:=|V|$ denotes the number of vertices
and~$m:=|E|$ denotes the number of edges in~$G$\fi. For a vertex~$v\in V$, we denote
by~$N_G(v):=\{u\in V\mid \{u,v\}\in E\}$ the \emph{open neighborhood of~$v$} and
  by~$N_G[v]:=N_G(v)\cup \{v\}$ the \emph{closed neighborhood of~$v$}. For a given set~$V' \subseteq V$, we define~$N_G(V'):= \bigcup_{v \in V'} N_G(v)$ as the \emph{neighborhood of~$V'$}. Moreover, let~$\deg_G(v):=|N(v)|$ be the \emph{degree} of a vertex~$v$ in~$G$ and~$\Delta_G:=\max_{v \in V} \deg_G(v)$ denote the \emph{maximum degree} of~$G$.
  For any two vertex sets~$V_1,V_2 \subseteq V$, we let~$E_G(V_1,V_2) := \lbrace \lbrace v_1, v_2 \rbrace \in E \mid v_1 \in V_1, v_2 \in V_2 \rbrace$ denote the set of edges between~$V_1$ and~$V_2$. For any vertex set~$V' \subseteq V$, we
let~$E_G(V') := E_G(V',V')$ denote the set of edges between the vertices of~$V'$. The \emph{subgraph induced by a vertex
    set}~$S$ is denoted by~$G[S]:=(S,E_G(S))$. For a given vertex set~$V' \subseteq V$, we let~$G-V':=G[V \setminus V']$ denote the graph that we obtain after deleting the vertices of~$V'$ from~$G$.  We may omit
  the subscript~$G$ if the graph is clear from the context.

A \emph{finite sequence}~$A=(a_0, a_1, \dots, a_{r-1})$ of length~$r \in \mathbb{N}_0$ is an~$r$-tuple of specific elements~$a_i$ (for example vertices or numbers). For given~$j \in \{0, \dots, r-1\}$, we refer to the~$j$th element of a finite sequence~$A$ as~$A(j)$.
    A \emph{path}~$P=(v_0,\dots,v_{r-1})$ is a finite sequence of vertices~$v_0,\dots,v_{r-1} \in V$, where~$\{v_i,v_{i+1}\} \in E$ for all~$i \in \{0,\dots,r-2\}$. A path~$P$ is called \emph{vertex-simple}, if no vertex appears twice on~$P$. A path is called \emph{edge-simple}, if there are no distinct~$i,j \in \{0,\dots,r-2\}$ such that~$\{P(i),P(i+1)\}=\{P(j),P(j+1)\}$. For a given path~$P=(P(0), \dots, P(r-1))$ we define the sets $V(P):=\{P(j) \mid j \in \{0,\dots, r-1\}\}$ and~$E(P) := \{ \{P(j),P(j+1)\} \mid j \in \{0,\dots,r-2\}\}$ as the set of vertices or edges on~$P$.

For the standard definitions of parameterized complexity refer
to~\cite{FKLMPPS15}.


  \subparagraph{Problem Definitions.} We now formally define the two main problems considered in
  this work, \textsc{ECS} and \textsc{Multi-STC}, as well as their extensions to input graphs
  with edge lists.
  
  \begin{definition} A $c$-\emph{colored labeling}~$L=(S^1_L,\ldots,S^c_L,W_L)$ of an undirected
    graph~$G=(V,E)$ is a partition of the edge set~$E$ into~$c+1$ color classes. The edges
    in~$S^i_L$,~$i\in \{1,\dots,c\}$, are \emph{strong} and the edges in~$W_L$ are \emph{weak}.
\begin{enumerate}
\item[1.] A~$c$-\emph{colored labeling}~$L$ is a \emph{proper labeling} if there exists no pair of edges~$e_1,e_2 \in S^i_L$ for some strong color~$i$, such that~$e_1 \cap e_2 \neq \emptyset$.
\item[2.] A~$c$-\emph{colored labeling}~$L$ is an \emph{STC-labeling} if there exists no pair of edges $\{ u, v \} \in S^i_L$ and
    $\{ v, w\} \in S^i_L$ such that $\{ u, w \}\not \in E$.
    \end{enumerate}
\end{definition}
We consider the following two problems.
\begin{quote}
  \textsc{Edge-Colorable Subgraph (ECS)}\\
  \textbf{Input}: An undirected graph~$G=(V, E)$ and
  integers~$c \in \mathds{N}$ and~$k \in \mathds{N}$.\\
  \textbf{Question}: Is there a $c$-colored proper labeling~$L$ with~$|W_L| \leq k$?
\end{quote}

\begin{quote}
  \textsc{Multi Strong Triadic Closure (Multi-STC)}\\
  \textbf{Input}: An undirected graph~$G=(V, E)$ and
  integers~$c \in \mathds{N}$ and~$k \in \mathds{N}$.\\
  \textbf{Question}: Is there a $c$-colored STC-labeling~$L$ with~$|W_L| \leq k$?
\end{quote}
If~$c$ is clear from the context, we may call a~$c$-colored labeling just \emph{labeling}.
  Two labelings~$L=(S^1_{L},\dots,S^c_{L},W_L)$, and~$L'=(S^1_{L'},\dots,S^c_{L'},W_{L'})$ for the same graph~$G=(V,E)$ are called \emph{partially equal} on a set~$E' \subseteq E$ if and only if for all~$e \in E'$ and~$i \in \{1,\dots,c\}$ it holds that~$e \in S^i_L \Leftrightarrow e \in S^i_{L'}$. If two labelings~$L$ and~$L'$ are partially equal on~$E'$ we write~$L|_{E'}=L'|_{E'}$. For given path~$P=(P(0), \dots, P(r-1))$ and labeling~$L=(S^1_L, \dots, S^c_L, W_L)$, we define the \emph{color sequence}~$Q_L^P$ of~$P$ under~$L$ as a finite sequence~$Q_L^P=(q_0,q_1,\dots,q_{r-2})$ of elements in~$\{0,\dots,c\}$, such that~$\{P(i), P(i+1)\} \in S_L^{q_i}$ if~$q_i \geq 1$ and~$\{P(i), P(i+1)\} \in W_L$ if~$q_i=0$. Throughout this work we call a $c$-colored STC-labeling $L$ (or proper labeling, respectively) \emph{optimal \iflong(for a graph $G$) \fi}if the number of weak edges $|W_L|$ is
minimal.

\subparagraph{Edge-Deletion Distance to Low-Degree Graphs and Component Order Connectivity.} We consider parameters related to the edge deletion-distance~$\pe_t$ to low-degree graphs and the vertex-deletion distance~$\lambda_t$ to graphs with small connected components\iflong ; they are formally defined as follows\fi.

First, we define the parameter~$\pe_t$. For a given graph~$G=(V,E)$ and a constant~$t \in \mathds{N}$, we call~$D_t \subseteq E$ an \emph{edge-deletion set of~$G$ and~$t$} if  the graph~$(V,E\setminus D_t)$ has maximum degree~$t$. We define the parameter~$\pe_t$ as the size of the minimum edge-deletion set of~$G$ and~$t$. Note that an edge-deletion set of~$G$ and~$t$ of size~$\pe_t$ can be computed in polynomial time~\cite{Gab83}. More importantly for our applications, we can compute a~$2$-approximation~$D'_t$ for an edge-deletion set of size~$\pe_t$ in linear time as follows: Add for each vertex~$v$ of degree at least~$t+1$ an arbitrary set of~$\deg(v)-t$ incident edges to~$D'_t$. Then~$|D'_t|\le \sum_{v\in V}\max (\deg(v)-t,0)$. This implies that~$D'_t$ is a 2-approximation since~$\sum_{v\in V}\max (\deg(v)-t,0) \leq 2\pe_t $ as every edge deletion decreases the degree of at most two vertices. 
A given edge-deletion set~$D_t$ induces the following important partition of the vertex set~$V$ of a graph.

\begin{definition}
Let~$t \in \mathds{N}$, let $G=(V,E)$ be a graph, and let~$D_t\subseteq E$ be an edge-deletion set of~$G$ and~$t$.
We call~$\mathscr{C}=\mathscr{C}(D_t):=\{v \in V \mid \exists e \in D_t: v \in e\}$ the set of \emph{core vertices} and~$\mathscr{P}=\mathscr{P}(D_t):=V \setminus \core$ the set of \emph{periphery vertices} of~$G$. 
\end{definition}

Note that for arbitrary~$t \in \mathds{N}$ and~$G$ we have~$|\mathscr{C}|\leq 2|D_t|$ and for every~$v \in \mathscr{P}$ it holds that~$\deg_G(v)\leq t$. Moreover, every vertex in~$\mathscr{C}$ is incident with at most~$t$ edges in~$E \setminus D_t$. In context of \textsc{ECS} and \textsc{Multi-STC}, for a given instance~$(G,c,k)$ we consider some fixed edge deletion set~$D_{t}$ of the input graph~$G$ and some integer~$t$ which depends on the value of~$c$.

Second, we define the parameter~$\lambda_t$. For a given graph~$G=(V,E)$ and a constant~$t \in \mathds{N}$, we call~$D \subseteq V$ an \emph{order-$t$ \cc} if every connected component in~$G-D$ contains at most~$t$ vertices. Then, we define the \emph{component order connectivity}~$\lambda_t$ to be the size of a minimum~oder-$t$ \cc. In context of~\textsc{ECS} we study~$\pc$, for the amount of colors~$c$. A $(c+1)$-approximation of the minimal~order-$c$-\cc{} can be computed in polynomial time~\cite{KL16}.

Note that the parameters are incomparable in the following sense: In a path~$P_n$ the parameter~$\pc$ can be arbitrary large when~$n$ increases while~$\pe_{c-1}=0$ for all~$c\geq 3$. In a star~$S_n$ the parameter~$\pe_{c-1}$ can be arbitrary large when~$n$ increases while~$\pc = 1$. 



\section{Problem Kernelizations for Edge-Colorable Subgraph}\label{sec:ed}

In this section, we provide problem kernels for~\textsc{ECS} parameterized by the edge deletion distance~$\pe_{c-1}$ to graphs with maximum degree~$c-1$, and the size~$\pc$ of a minimum order-$c$ component cover. We first show that~\textsc{ECS} admits a kernel with~$\mathcal{O}(\pe_{c-1}\cdot c)$ vertices and~$\mathcal{O}(\pe_{c-1}\cdot c^2)$ edges that can be computed in~$\mathcal{O}(n+m)$~time. Afterwards, we consider~$\pc$ and show that~\textsc{ECS} admits a problem kernel with~$\Oh(c^3\pc)$ vertices, which is a linear vertex kernel for every fixed value of~$c$. Note that if~$c=1$ we can solve \textsc{ECS} by computing a maximal matching in polynomial time. Hence, we assume~$c\geq 2$ for the rest of this section. In this case the problem is NP-hard~\cite{FOW02}.

\subsection{Edge Deletion-Distance to Low-Degree~Graphs}
\label{sec:ecsxi}
The kernelization presented inhere is based on Vizing's Theorem~\cite{V64}. Note that Vizing's Theorem implies, that an ECS instance~$(G,c,k)$ is always a yes-instance if~$\pe_{c-1}=0$. Our kernelization relies on the following lemma. This lemma is a reformulation of a known fact about edge colorings~\cite[Theorem 2.3]{S12} which, in turn, is based on the so-called \emph{Vizing Fan Equation}~\cite[Theorem 2.1]{S12}.

\begin{lemma} \label{Lemma: Fan Equation}
Let~$G=(V,E)$ be a graph and let~$e:=\{u,v\} \in E$. Moreover, let~$c:=\Delta_{G}$ and let~$L$ be a proper~$c$-colored labeling for the graph~$(V,E \setminus \{e\})$ such that~$W_L =\emptyset$. If for all~$Z \subseteq N_G(u)$ with~$|Z| \geq 2$ and~$v \in Z$ it holds that~$\sum_{z \in Z} (\deg_{G}(z) + 1 - c) < 2$, then there exists a proper~$c$-colored labeling~$L'$ for~$G$ such that~$W_{L'}=\emptyset$.
%
%
\end{lemma} 

We now use Lemma~\ref{Lemma: Fan Equation} as a plug-in for \textsc{ECS} to prove the next lemma which is the main tool that we need for our kernelization. In the proof, we exploit the fact that, given any proper labeling~$L$ for a graph~$G=(V,E)$, the labeling~$(S^1_L, \dots, S^c_L, \emptyset)$ is a proper labeling for the graph~$(V,E \setminus W_L)$. 

\begin{lemma}\label{Lemma: Add edges of low degree}
Let~$L:=(S^1_L, S^2_L, \dots, S^c_L, W_L)$ be a proper labeling with~$|W_L|=k$ for a graph~$G:=(V,E)$. Moreover, let~$e:=\{u,v\} \subseteq V$ such that~$e \not \in E$ and let~$G':=(V,E \cup \{e\})$ be obtained from~$G$ by adding~$e$. If for one endpoint~$u \in e$ it holds that every vertex~$w \in N_{G'}[u]$ has degree at most~$c-1$ in~$G'$, then there exists a proper labeling~$L'$ for~$G'$ with~$|W_{L'}| = k$.
\end{lemma}

\begin{proof}
Consider the auxiliary graph~$G_{\text{aux}}:=(V,E \setminus W_L)$. Since~$L$ is a proper labeling for~$G$, we conclude that~$L_{\text{aux}}:=(S^1_L, \dots, S^c_L, \emptyset)$ is a proper labeling for~$G_{\text{aux}}$. Let~$H_{\text{aux}}:=(V, E_H)$ where~$E_H:=(E \setminus W_L) \cup \{e\}$. In order to prove the lemma, we show that there exists a proper~labeling~$L'_{\text{aux}}$ for~$H_{\text{aux}}$ such that~$W_{L'_{\text{aux}}}=\emptyset$. 

To this end, we first consider the maximum degree of~$H_{\text{aux}}$. Observe that~$\deg_{H_{\text{aux}}}(w) \leq \deg_{G'}(w)$ for all~$w \in V$. Hence, the property that~$\deg_{G'}(w) \leq c-1$ for all~$w \in N_{G'}[u]$ implies~$\Delta_{H_\text{aux}}= \max(\Delta_{G_\text{aux}},c-1)$. Since~$L_\text{aux}$ is a proper~$c$-colored labeling for~$G_\text{aux}$ we know that~$\Delta_{G_{\text{aux}}} \leq c$ and therefore we have~$\Delta_{H_{\text{aux}}} \leq c$. So, to find a proper~$c$-colored labeling without weak edges for~$H_{\text{aux}}$ it suffices to consider the following cases.


\textbf{Case 1:} $\Delta_{H_\text{aux}}\leq c-1$\textbf{.} Then, there exists a proper labeling~$L'_{\text{aux}}$ for~$H_\text{aux}$ such that~$W_{L'_{\text{aux}}}=\emptyset$ due to Vizing's Theorem.

\textbf{Case 2:} $\Delta_{H_\text{aux}} = c$\textbf{.} In this case we can apply Lemma~\ref{Lemma: Fan Equation}: Observe that~$(V,E_H \setminus \{e\}) = G_\text{aux}$ and~$L_{\text{aux}}$ is a proper labeling for~$G_\text{aux}$ such that~$W_{L_{\text{aux}}}=\emptyset$. Consider an arbitrary~$Z \subseteq N_{H_\text{aux}}(u)$ with~$|Z| \geq 2$ and~$v \in Z$. Note that~$Z \subseteq N_{H_\text{aux}}(u)$ implies~$\deg_{H_\text{aux}}(z) \leq c-1$ for all~$z \in Z$. It follows that~$\sum_{z \in Z} (\deg_{H_{\text{aux}}}(z) + 1 - c) < 2$. Since~$Z$ was arbitrary, Lemma~\ref{Lemma: Fan Equation} implies that there exists a proper labeling~$L'_{\text{aux}}$ for~$H_\text{aux}$ such that~$W_{L'_{\text{aux}}}= \emptyset$.

%

We now define~$L':=(S^1_{L'_\text{aux}}, S^2_{L'_\text{aux}}, \dots S^c_{L'_\text{aux}}, W_L)$. Note that the edge set~$E \cup \{e\}$ of~$G'$ can be partitioned into~$W_L$ and the edges of~$G'_{\text{aux}}$. Together with the fact that~$L'_{\text{aux}}$ is a labeling for~$G'_{\text{aux}}$ it follows that every edge of~$G'$ belongs to exactly one color class of~$L'$. Moreover, it obviously holds that~$|W_{L'}|=|W_L|=k$. Since there is no vertex with two incident edges in the same strong color class~$S^i_{L'_\text{aux}}$, the labeling~$L'$ is a proper labeling for~$G'$. 
\end{proof}

We now introduce the kernelization rule. Recall that~$\core$ is the set of vertices that are incident with at least one of the~$\pe_{c-1}$ edge-deletions that transform~$G$ into a graph with maximum degree~$c-1$. We make use of the fact that edges that have at least one endpoint~$u$ that is not in~$\core \cup N(\core)$ satisfy~$\deg(w) \leq c-1$ for all~$w \in N[u]$. Lemma~\ref{Lemma: Add edges of low degree} guarantees that these edges are not important to solve an instance of~\textsc{ECS}.

\begin{redrule} \label{Rule: Kernelrule ECS}
Remove all vertices in~$V \setminus (\core \cup N(\core))$ from~$G$.
\end{redrule}

\begin{proposition}
Rule~\ref{Rule: Kernelrule ECS} is safe.
\end{proposition}

\iflong
\begin{proof}
Let~$(G'=(V',E'),c,k)$ be the reduced instance after applying Rule~\ref{Rule: Kernelrule ECS}. We prove the safeness of Rule~\ref{Rule: Kernelrule ECS} by showing that~there is a proper labeling with at most~$k$ weak edges for~$G$ if and only if there is a proper labeling with~$k$ weak edges for~$G'$.

$(\Rightarrow)$ Let~$L=(S^1_L, S^2_L, \dots, S^c_L, W_L)$ be a proper labeling with~$|W_L| \leq k$ for~$G$. Then, obviously~$L':=(S^1_L \cap E', S^2_L \cap E', \dots, S^c_L \cap E', W_L \cap E')$ is a proper labeling for~$G'$ with~$|W_{L'}| \leq |W_L| \leq k$.

$(\Leftarrow)$ Conversely, let~$L'=(S^1_{L'}, S^2_{L'}, \dots, S^c_{L'}, W_{L'})$ be a proper labeling with~$|W_{L'}| \leq k$ for~$G'$. Let~$E \setminus E'= \{e_1, e_2, \dots, e_p\}$. We define~$p+1$ graphs~$G_0, G_1, G_2, \dots, G_p$ by~$G_0:=(V,E')$, and~$G_i:=(V,E' \cup \{e_1, \dots, e_i\})$ for~$i \in \{1,\dots,p\}$. Note that~$G_p =G$, $\deg_{G_i}(v) \leq \deg_G(v)$, and~$N_{G_i}(v) \subseteq N_G(v)$ for every~$i \in \{0, 1, \dots, p\}$, and~$v \in V$. We prove by induction over~$i$ that all~$G_i$ have a proper labeling with at most~$k$ weak edges.

\textit{Base Case:} $i=0$. Then, since~$G_0$ and~$G'$ have the exact same edges,~$L'$ is a proper labeling for~$G_0$ with at most~$k$ weak edges.

\textit{Inductive Step:} $0< i \leq p$. Then, by the inductive hypothesis, there exists a proper labeling~$L_{i-1}$ for~$G_{i-1}=(V, E' \cup \{e_1, \dots, e_{i-1}\})$ with at most~$k$ weak edges. From~$E' = E( \core \cup N(\core))$ we conclude~$e_i \in E \setminus E(\core \cup N(\core)) = E(\peri) \setminus E(N(\core))$. Hence, for at least one of the endpoints~$u$ of~$e$ it holds that~$N_G[u] \subseteq \peri$. Therefore~$\deg_G(w) \leq c-1$ for all~$w \in N_G[u]$. Together with the facts that~$\deg_{G_i}(w) \leq \deg_{G}(w)$ and~$N_{G_i}(w) \subseteq N_G(w)$ we conclude~$\deg_{G_i}(w) \leq c-1$ for all~$w \in N_{G_i} [u]$. Then, by Lemma~\ref{Lemma: Add edges of low degree}, there exists a proper labeling~$L_i$ for~$G_i$ such that~$|W_{L_i}|=|W_{L_{i-1}}| \leq k$.
\end{proof}
\fi

\iflong It remains to state the kernel result.\fi

\xiecs*
\iflong
\begin{proof}
Let~$(G,c,k)$ be an instance of \textsc{ECS}. We apply Rule~\ref{Rule: Kernelrule ECS} on~$(G,c,k)$ as follows: First, we compute a 2-approximation~$D'_{c-1}$ of the smallest possible edge-deletion set~$D_{c-1}$ in~$\Oh(n+m)$ time~as described in Section~\ref{sec:prelim}. Let~$\core:=\core(D'_{c-1})$ and note that~$|D'_{c-1}| \leq 2 \pe_{c-1}$. We then remove all vertices in~$V \setminus (\core \cup N_G(\core))$ from~$G$ which can also be done in~$\Oh(n+m)$ time. Hence, applying Rule~\ref{Rule: Kernelrule ECS} can be done in~$\Oh(n+m)$ time.

We next show that after this application of Rule~\ref{Rule: Kernelrule ECS} the graph consists of at most~$4 \pe_{c-1}\cdot c$ vertices and~$\Oh( \pe_{c-1}\cdot c^2)$ edges. Since~$D'_{c-1}$ is a 2-approximation of the smallest possible edge-deletion set we have~$|\core| \leq 4\pe_{c-1}$. Since every vertex in~$\core$ has at most~$c-1$ neighbors in~$V \setminus \core$, we conclude~$|\core \cup N(\core)| \leq 4 \pe_{c-1} \cdot c$. In~$E(\core \cup N(\core))$ there are obviously the at most~$4\pe_{c-1}$ edges of~$D'_{c-1}$. Moreover, each of the at most~$4\pe_{c-1} \cdot c$ vertices might have up to~$c-1$ incident edges. Hence, after applying Rule~\ref{Rule: Kernelrule ECS}, the reduced instance has~$\Oh(\xi_{c-1} \cdot c^2)$~edges.
\end{proof}
\fi

If we consider~\textsc{Edge Coloring} instead of~\textsc{ECS}, we can immediately reject if one vertex has degree more than~$c$. Then, since there are at most~$|\core|\leq 2 \pe_{c-1}$ vertices that have a degree of at least~$c$, Theorem~\ref{Theorem: ECS xi kernel} implies the following.

\begin{corollary} \label{Corollary: ECS hc kernel}
Let~$h_c$ be the number of vertices with degree~$c$. \textsc{Edge Coloring} admits a problem kernel with~$\Oh(h_c \cdot c)$ vertices and~$\Oh(h_c \cdot c^2)$ edges that can be computed in~$\Oh(n+m)$~time.
\end{corollary}

\subsection{Component Order Connectivity}
\label{sec:vc}
In this section we present a problem kernel for~\textsc{ECS} parameterized by the number of strong colors~$c$ and the component order connectivity~$\pc$. We prove that~\textsc{ECS} admits a problem kernel with~$\Oh(c^3 \cdot \pc)$ vertices, which is a linear vertex kernel for every fixed value of~$c$. Our kernelization is based on the Expansion Lemma\iflong{}~\cite{R05}, a generalization of the Crown Rule~\cite{CFJ04}. We use the formulation given by Cygan~et~al.\fi~\cite{FKLMPPS15}. \iflong
\begin{lemma}[Expansion Lemma] \label{Lemma: Expansion Lemma}

 Let~$q$ be a positive integer and~$G$ be a bipartite graph with partite sets~$A$ and~$B$ such that~$|B| \geq q|A|$ and there are no isolated vertices in~$B$. Then there exist nonempty vertex sets~$X \subseteq A$ and~$Y \subseteq B$ with~$N(Y) \subseteq X$. Moreover, there exist edges~$M \subseteq E(X,Y)$ such that
\begin{enumerate}
\item[a)] every vertex of~$X$ is incident with exactly~$q$ edges of~$M$, and
\item[b)] $q \cdot |X|$ vertices in~$Y$ are endpoints of edges in~$M$.
\end{enumerate}
The sets~$X$ and $Y$ can be found in polynomial time.
\end{lemma}
\fi
To apply the Expansion Lemma on an instance of~\textsc{ECS}, we need the following definition for technical reasons.

\begin{definition}
For a given graph~$G=(V,E)$, let~$D$ be an~order-$c$ \cc. We say that~$D$ is \emph{saturated} if for every~$v \in D$ it holds that~$E_G(\{v\}, V\setminus D) \neq \emptyset$.
\end{definition}

Note that every order-$c$ \cc{}~$D'$ can be transformed into a saturated~order-$c$ \cc{} by removing any vertex~$v \in D'$ with~$N(v) \subseteq D'$ from~$D'$ while such a vertex exists. Let~$(G=(V,E),c,k)$ be an instance of \textsc{ECS} and let~$D \subseteq V$ be a saturated order-$c$ \cc. Furthermore, let~$I:=V \setminus D$ be the remaining set of vertices. \iflong Consider the following simple reduction rule. \fi

\begin{redrule} \label{Rule: Remove isolated vertices}
If there exists a set~$J \subseteq I$ such that~$J$ is a connected component in~$G$, remove all vertices in~$J$ from~$G$.
\end{redrule}

Rule~\ref{Rule: Remove isolated vertices} is safe since~$|J| \leq c$ and therefore the graph~$G[J]$ has maximum degree~$c-1$ and can be labeled by Vizing's Theorem with~$c$ colors.  For the rest of this section we assume that~$(G,c,k)$ is reduced regarding Rule~\ref{Rule: Remove isolated vertices}. The following proposition is a direct consequence of the Expansion Lemma.

\begin{proposition} \label{Prop: ExpLemma applied}
Let~$(G=(V,E),c,k)$ be an instance of \textsc{ECS} that is reduced regarding Rule~\ref{Rule: Remove isolated vertices}, let~$D$ be a saturated~order-$c$ \cc{} of~$G$, and let~$I:=V \setminus D$. If~$|I| \geq c^2 \cdot |D|$, then there exist nonempty sets~$X \subseteq D$ and~$Y \subseteq I$ with~$N(Y) \subseteq X \cup Y$. Moreover, there exists a set~$M \subseteq E(X,Y)$ such that
\begin{enumerate}
\item[a)] every vertex of~$X$ is incident with exactly~$c$ edges of~$M$, and
\item[b)] $c \cdot |X|$ vertices in~$Y$ are endpoints of edges in~$M$ and every connected component in~$G[Y]$ contains at most one such vertex.
\end{enumerate}
The sets~$X$ and~$Y$ can be computed in polynomial time.
\end{proposition}
\iflong
\begin{proof}
We prove the proposition by applying Lemma~\ref{Lemma: Expansion Lemma}. To this end we define an equivalence relation~$\sim$ on the vertices of~$I$: Two vertices~$v,u \in I$ are equivalent, denoted~$u \sim v$ if and only if~$u$ and~$v$ belong to the same connected component in~$G[I]$. Obviously,~$\sim$ is an equivalence relation. For a given vertex~$u \in I$, let~$[u]:= \{v \in I \mid v \sim u\}$ denote the equivalence class of~$u$. Note that~$|[u]| \leq c$ since~$D$ is an order-$c$ \cc.

We next define the auxiliary graph~$G_\aux$, on which we will apply Lemma~\ref{Lemma: Expansion Lemma}. Intuitively, we obtain~$G_\aux$ from~$G$ by deleting all edges in~$E_G(D)$ and merging the at most~$c$ vertices in every equivalence class in~$I$. Formally~$G_\aux:=(D \cup I^*, E_\aux)$, with~$I^* := \{ [u] \mid u \in I \}$ and
\begin{align*}
E_\aux &:= \{ \{[u], v\} \mid [u] \in I^*, v \in \bigcup_{w \in [u]} (N_G(w) \setminus I)\}.
\end{align*}
Note that~$G_\aux$ can be computed from~$G$ in polynomial time and that~$|I| \geq |I^*| \geq \frac{1}{c} |I|$.

Observe that~$G_\aux$ is bipartite with partite sets~$D$ and~$I^*$. Since~$G$ is reduced regarding Rule~\ref{Rule: Remove isolated vertices}, every~$[u] \in I^*$ is adjacent to some~$v \in D$ in~$G_\aux$. Furthermore, since~$D$ is saturated, every~$v \in D$ is adjacent to some~$u \in I$ in~$G$ and therefore~$\{v,[u]\} \in E_\aux$. Hence, $G_\aux$ is a bipartite graph without isolated vertices. Moreover, from~$|I| \geq c^2 |D|$ and~$|I^*| \geq \frac{1}{c} |I|$ we conclude~$|I^*| \geq c \cdot |D|$. By applying Lemma~\ref{Lemma: Expansion Lemma} on~$G_\aux$ we conclude that there exist nonempty vertex sets~$X' \subseteq D$ and~$Y' \subseteq I^*$ with~$N_{G_\aux}(Y') \subseteq X'$ that can be computed in polynomial time such that there exists a set~$M' \subseteq E_{G_\aux}(X',Y')$ of edges, such that every vertex of~$X'$ is incident with exactly~$c$ edges of~$M'$, and~$c \cdot |X'|$ vertices in~$Y'$ are endpoints of edges in~$M'$.

We now describe how to construct the sets~$X$, $Y$, and~$M$ from~$X'$, $Y'$, and~$M'$. We set~$X:= X' \subseteq D$, and~$Y:= \bigcup_{[u] \in Y'} [u] \subseteq I$.  We prove that~$N_G(Y) \subseteq X \cup Y$. Let~$y \in Y$. Note that all neighbors of~$y$ in~$I$ are elements of~$Y$ by the definition of the equivalence relation~$\sim$ and therefore
\begin{align*}
N_G(y) \subseteq N_{G_\aux}([y]) \cup Y \subseteq X' \cup Y = X \cup Y.
\end{align*}

Next, we construct~$M \subseteq E_G(X,Y)$ from~$M'$. To this end we define a mapping~$\pi: M' \rightarrow E_G(X,Y)$. For every edge~$\{[u],v\} \in M'$ with~$[u] \in Y'$ and~$v \in X'$ we define~$\pi(\{[u],v\}):=\{w,v\}$, where~$w$ is some fixed vertex in~$[u]$. We set~$M:= \{\pi(e') \mid e' \in M'\}$. It remains to show that the statements~$a)$ and~$b)$ hold for~$M$. 

$a)$ Observe that~$\pi(\{[u_1],v_1\})= \pi(\{[u_2],v_2\})$ implies~$[u_1]=[u_2]$ and~$v_1=v_2$ and therefore, the mapping~$\pi$ is injective. We conclude~$|M|=|M'|$. Moreover, observe that the edges of~$M$ have the same endpoints in~$X$ as the edges of~$M'$. Thus, since every vertex of~$X'$ is incident with exactly~$c$ edges of~$M'$ it follows that statement~$a)$ holds for~$M$.

$b)$ By the conditions~$a)$ and~$b)$ of Lemma~\ref{Lemma: Expansion Lemma}, no two edges in~$M'$ have a common endpoint in~$Y'$. Hence, in every connected component in~$G[Y]$ there is at most one vertex incident with an edge in~$M$. Moreover, since~$|M|=|M'|$ and there are exactly~$c \cdot |X'|$ vertices in~$Y'$ that are endpoints of edges in~$M'$ we conclude that statement~$b)$ holds for~$M$.
\end{proof}
The following rule is the key rule for our kernelization.
\fi

\begin{redrule} \label{Rule: ExpLemma Rule}
If~$|I| \geq c^2 \cdot |D|$, then compute the sets~$X$ and~$Y$ from Proposition~\ref{Prop: ExpLemma applied}, delete all vertices in~$X \cup Y$ from~$G$, and decrease~$k$ by~$|E_G(X, V)| - c \cdot |X|$.
\end{redrule}

\begin{proposition}
Rule~\ref{Rule: ExpLemma Rule} is safe.
\end{proposition}

\iflong
\begin{proof}
Let~$G'= (V',E'):= G- (X \cup Y)$ be the graph after applying Rule~\ref{Rule: ExpLemma Rule}. Note that~$V' = V \setminus (X \cup Y)$, and~$E'=E \setminus (E_G(X \cup Y ,V))$. Moreover, let~$k':= k - |E_G(X,V)| + c \cdot |X|$. We show that there exists a proper labeling~$L$ with~$|W_L| \leq k$ for~$G$ if and only if there is a proper labeling~$L'$ with~$|W_{L'}| \leq k'$ for~$G'$.

$(\Rightarrow)$ Let~$L$ be a proper labeling for~$G$ with~$|W_L| \leq k$. We define a labeling~$L'$ for~$G'$ by~$L':= (S^1_L \cap E' , S^2_L \cap E', \dots, S^c_L \cap E', W_L \cap E')$. Obviously, no vertex in $V$ is incident with two edges of the same strong color under~$L$, and therefore no vertex in~$V' \subseteq V$ is incident with two edges of the same strong color under~$L'$. Hence,~$L'$ is a proper labeling for~$G'$. It remains to show that~$|W_{L'}| \leq k'$. Obviously, every vertex~$x \in X$ is incident with at most~$c$ edges of distinct strong colors under~$L$, since~$L$ is a proper labeling. Hence, the maximum number of strong edges in~$E_G(X,V)$ is~$c \cdot |X|$. Thus, we have~$|W_L \cap E_G(X,V)| \geq |E_G(X,V)| - c \cdot |X|$. Therefore,
\begin{align*}
|W_{L'}| &= |W_L \cap E'| = |W_L| - |W_L \cap E_G(X \cup Y,V)| \\
		 &\leq |W_L| - |W_L \cap E_G(X,V)| \leq k -|E_G(X,V)| + c \cdot |X| = k'.
\end{align*}

$(\Leftarrow)$ Conversely, let~$L'$ be a proper labeling for~$G'$ with~$|W_{L'}| \leq k'$. We now describe how to construct a labeling~$L$ for~$G$ with~$|W_L| \leq k$ from~$L'$. We set~$W_L := W_{L'} \cup (E_G(X,V) \setminus M)$. This implies \iflong
\begin{align*}
|W_L| = |W_{L'}| + |E_G(X,V)| - |M| \leq k' + |E_G(X,V)| - c|X| = k.
\end{align*}
\else
$|W_L| = |W_{L'}| + |E_G(X,V)| - |M| \leq k' + |E_G(X,V)| - c|X| = k$.
\fi
Next, we describe to which strong color classes of~$L$ we add the remaining edges of~$G$. Since~$N_G(Y) \subseteq X \cup Y$ it remains to label all edges in~$E'\setminus W_{L'} \cup E_G(Y) \cup M$.

First, consider the edges in~$E' \setminus W_{L'}$. Every edge~$e \in E' \setminus W_{L'}$ has a strong color~$i$ under~$L'$. We then add~$e$ to~$S^i_L$. Note that this implies~$L|_{E'} = L'|_{E'}$.

Second, consider the edges in~$M$. For each~$x \in X$ we define a set~$B_x^M:= \{e \in M \mid x \in e\} \subseteq M$. By Proposition~\ref{Prop: ExpLemma applied}~a), every~$x \in X$ is incident with exactly~$c$ edges in~$M$. Thus, ~$|B_x^M|=c$ and we let~$e^1_x, e^2_x, \dots, e^c_x$ denote be the elements of~$B_x^M$. For every~$x \in X$ we add~$e^i_x$ to the strong color class~$S^i_L$. Note that by Proposition~\ref{Prop: ExpLemma applied}, there are~$c|X|$ vertices in~$Y$ that are incident with edges in~$M$. Hence, the family~$\{ B_x^M \subseteq M \mid x \in X\}$ forms a partition of~$M$ and therefore every edge in~$M$ belongs to exactly one strong color class of~$L$.

Finally, consider the edges in~$E_G(Y)$. Let~$J \subseteq Y$ be a connected component in~$G[Y]$. Note that~$N_G(J) \subseteq J \cup X$ and observe that by Proposition~\ref{Prop: ExpLemma applied} b) there is at most one vertex~$v \in J$ that is an endpoint of some edge in~$M$. Hence, there is at most one edge in~$E_G(J,X)$ that belongs to some strong color class~$S^i_L$. Since~$D$ is an~order-$c$~\cc, we know that~$|J| \leq c$, and therefore~$\Delta_{G[J]} \leq c-1$. Consequently, there exists a proper labeling~$L''=(S^1_{L''},\dots, S^c_{L''},W_{L''})$ for~$G[J]$ due to Vizing's Theorem. Without loss of generality we can assume that~$v$ is not incident with an edge in~$S^{i}_{L''}$: If there exists an edge~$\{v,w\} \in S^{i}_{L''}$, there exists one strong color class~$S^{j}_{L''}$ that contains no edge incident with~$v$ since~$\deg_{G[J]}(v) \leq c-1$ and we simply interchange the edges in~$S^i_{L''}$ and~$S^j_{L''}$. Then, for every~$t \in \{1, \dots, c\}$ we add all edges in~$S^t_{L''}$ to the strong color class~$S^t_L$.
%

It remains to show that~$L$ is a proper labeling. To this end, we show for every vertex~$v \in V$, that~$v$ is not incident with two edges of the same strong color under~$L$. Consider the following case distinction.

\textbf{Case 1:}~$v \in V \setminus (X \cup Y)$\textbf{.} Then, since~$E_G(\{v\},Y) = \emptyset$, and~$E_G(\{v\},X) \subseteq W_L$, every strong edge incident with~$v$ has the same strong color under~$L$ as it has under~$L'$. Since~$L'$ is a proper labeling, the vertex~$v \in V\setminus (X \cup Y)$ is not incident with two edges of the same strong color under~$L$.

\textbf{Case 2:}~$v \in X$\textbf{.} Then, since~$E_G(\{v\},V \setminus Y) \subseteq W_L$, and~$E_G(\{v\},Y) \setminus B_v^M \subseteq W_L$, all strong edges incident with~$v$ are elements of~$B_v^M$. Since~$B_v^M=\{e^1_v, \dots, e^c_v\}$, and every~$e^i_v \in S^i_L$, the vertex~$v$ is not incident with two edges of the same strong color under~$L$. 

\textbf{Case 3:}~$v \in Y$\textbf{.} Let~$J \subseteq Y$ be a connected component in~$G[Y]$ such that~$v \in J$. Note that~$N_G(v) \subseteq X \cup J$. First, consider the case, that~$v$ has no strong neighbors in~$X$. Then, there are no ECS violations since~$L|_{E_G(J)}$ is a proper labeling for~$G[J]$ by Vizing's Theorem. Second, consider the case that~$v$ has strong neighbors in~$X$. Then, by Proposition~\ref{Prop: ExpLemma applied}, there is exactly one edge in~$E_G(\{v\},X)$ that belongs to~$M$ and therefore is in some strong color class~$S^i_L$. By the construction of~$L$, all edges in~$E_G(\{v\},J)$ have pairwise distinct strong colors which are all distinct from~$i$ under~$L$. Therefore, the vertex~$v$ is not incident with two edges of the same strong color under~$L$.
\end{proof}
\fi

Rules~\ref{Rule: Remove isolated vertices} and \ref{Rule: ExpLemma Rule} together with the fact that we can compute a $(c+1)$-approximation of the minimum~order-$c$~\cc{} in polynomial time~\cite{KL16} give us the following.

\vckernel*
\iflong
\begin{proof}
We first consider the running time. We use a $(c+1)$-approximation for the minimum~oder-$c$~\cc{} and compute an~order-$c$~\cc{}~$D'$ in polynomial time~\cite{KL16}. Afterwards we remove any vertex~$v\in D'$ with~$N(v) \subseteq D'$ from~$D'$ while such a vertex exists and we end up with a saturated~order-$c$~\cc~$D\subseteq D'$.
Afterwards, consider Rules~\ref{Rule: Remove isolated vertices} and~\ref{Rule: ExpLemma Rule}. Obviously, one application of Rule~\ref{Rule: Remove isolated vertices} can be done in polynomial time if~$D$ is known. Moreover, Rule~\ref{Rule: ExpLemma Rule} can also be applied in polynomial time due to Proposition~\ref{Prop: ExpLemma applied}. Since every application of one of these two rules removes some vertices, we can compute an instance that is reduced regarding Rules~\ref{Rule: Remove isolated vertices} and~\ref{Rule: ExpLemma Rule} from an arbitrary input instance of~\textsc{ECS} in polynomial time.

We next consider the size of a reduced instance~$(G=(V,E),c,k)$ of~\textsc{ECS} regarding Rules~\ref{Rule: Remove isolated vertices} and~\ref{Rule: ExpLemma Rule}. Let~$D \subseteq V$ be a~$(c+1)$-approximate saturated~order-$c$~\cc, and let~$I:=V\setminus D$. Since no further application of Rule~\ref{Rule: ExpLemma Rule} is possible, we conclude~$|I| < c^2 \cdot |D|$. Thus, we have~$|V|=|I|+|D|< (c^2 + 1) \cdot |D| \leq (c^2 + 1) \cdot (c+1) \cdot  \pc \in \Oh(c^3 \pc)$.
\end{proof}
\fi

\section{Multi-STC parameterized by Edge Deletion-Distance to Low-Degree Graphs}\label{sec:stc}

In this section we provide a problem kernelization for \textsc{Multi-STC} parameterized by~$\pe_{c-1}$ when~$c \leq 4$. Before we describe the problem kernel, we briefly show that \textsc{Multi-STC} does not admit a polynomial
kernel for the component order connectivity~$\pe_{c-1}$ even if~$c=1$: If~$\text{NP} \not \subseteq \text{coNP} / \text{poly}$, \textsc{STC} does not admit a polynomial kernel if parameterized by the number of strong edges~\cite{GK18} which---in nontrivial instances---is bigger than the size of a maximal matching~$M$. Since the vertex cover number~$s$ is never larger than~$2|M|$, this implies that \textsc{Multi-STC} has no polynomial kernel if parameterized by~$s$ unless~$\badstuffhappens$. Since~$\pc \leq s$, we conclude that \textsc{Multi-STC} does not admit a polynomial kernel for~$\pc$ unless~$\badstuffhappens$.

\todog{Neuer Zwischentext.} Next, consider parameterization by~$\pe_{c-1}$. Observe that Rule~\ref{Rule: Kernelrule ECS} which gives a problem kernel for~\textsc{ECS} does not work for~\textsc{Multi-STC}; see Figure~\ref{Figure: Counterexample ECS Rule for STC} for an example. Furthermore, for \textsc{Multi-STC} we need a fundamental new approach: For STC-labelings the maximum degree  and the number of colors are not as closely related as in \textsc{ECS}, and therefore, Lemma~\ref{Lemma: Fan Equation} might not be helpful for \textsc{Multi-STC}. Moreover, in the proof of Lemma~\ref{Lemma: Add edges of low degree} we exploit that in \textsc{ECS} we may remove weak edges from the instance, which does not hold for \textsc{Multi-STC} since removing a weak edge may produce~$P_3$s. However, the results for~\textsc{ECS} parameterized by~$(\pe_{c-1},c)$ can be lifted to the seemingly harder~\textsc{Multi-STC} for~$c \in \{1,2,3,4\}$. We will first discuss the cases~$c=1$ and~$c=2$. For the cases~$c \in \{3,4\}$ we  show the more general statement that~\textsc{Multi-STC} admits a problem kernel with~$\Oh(\pe_{\lfloor \frac{c}{2} \rfloor +1}\cdot c)$ vertices.

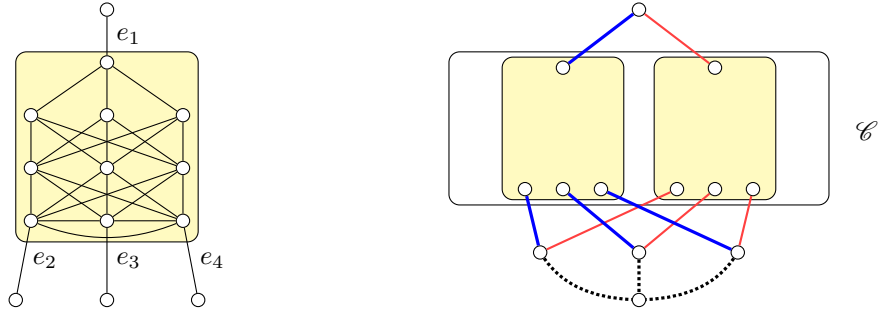
\begin{figure}
\begin{center}

\begin{tikzpicture}[xscale=1, yscale=0.7]
\tikzstyle{knoten}=[circle,fill=white,draw=black,minimum size=5pt,inner sep=0pt]
\tikzstyle{bez}=[inner sep=0pt]]

\begin{scope}[xshift=-7cm]
\draw[rounded corners, fill=yellow!30] (-1.2, 5.2) rectangle (1.2, 1.6) {};

\node[knoten] (g0) at (0,6) {};

\node[knoten] (g1) at (0,5) {};

\node[knoten] (g21) at (-1,4) {};
\node[knoten] (g22) at (0,4) {};
\node[knoten] (g23) at (1,4) {};

\node[knoten] (g31) at (-1,3) {};
\node[knoten] (g32) at (0,3) {};
\node[knoten] (g33) at (1,3) {};

\node[knoten] (g41) at (-1,2) {};
\node[knoten] (g42) at (0,2) {};
\node[knoten] (g43) at (1,2) {};

\node[knoten] (g51) at (-1.2,0.5) {};
\node[knoten] (g52) at (0,0.5) {};
\node[knoten] (g53) at (1.2,0.5) {};

\draw[-] (g0) to node [anchor=west]  {$e_1$} (g1);

\draw[-] (g1) to (g21);
\draw[-] (g1) to (g22);
\draw[-] (g1) to (g23);

\draw[-] (g21) to (g31);
\draw[-] (g21) to (g32);
\draw[-] (g21) to (g33);
\draw[-] (g22) to (g31);
\draw[-] (g22) to (g32);
\draw[-] (g22) to (g33);
\draw[-] (g23) to (g31);
\draw[-] (g23) to (g32);
\draw[-] (g23) to (g33);

\draw[-] (g31) to (g41);
\draw[-] (g31) to (g42);
\draw[-] (g31) to (g43);
\draw[-] (g32) to (g41);
\draw[-] (g32) to (g42);
\draw[-] (g32) to (g43);
\draw[-] (g33) to (g41);
\draw[-] (g33) to (g42);
\draw[-] (g33) to (g43);

\draw[-] (g41) to (g42);
\draw[-] (g42) to (g43);
\draw[-, bend left] (g43) to (g41);

\draw[-] (g41) to node [anchor=west]  {$e_2$} (g51);
\draw[-] (g42) to node [anchor=west]  {$e_3$} (g52);
\draw[-] (g43) to node [anchor=west]  {$e_4$} (g53);
\end{scope}

\draw[rounded corners] (-2.5, 5.2) rectangle (2.5, 2.3) {};
\node[knoten] (i0) at (0,6) {};

\begin{scope}[xshift=0.5cm,yshift=0.4cm]
\draw[rounded corners, fill=yellow!30] (-2.3, 4.7) rectangle (-0.7, 2) {};

\node[knoten] (l21) at (-1.5,4.5) {};
\node[knoten] (l31) at (-2,2.2) {};
\node[knoten] (l32) at (-1.5,2.2) {};
\node[knoten] (l33) at (-1,2.2) {};
\end{scope}

\begin{scope}[xshift=2.5cm,yshift=0.4cm]
\draw[rounded corners, fill=yellow!30] (-2.3, 4.7) rectangle (-0.7, 2) {};

\node[knoten] (r21) at (-1.5,4.5) {};
\node[knoten] (r31) at (-2,2.2) {};
\node[knoten] (r32) at (-1.5,2.2) {};
\node[knoten] (r33) at (-1,2.2) {};
\end{scope}

\node[knoten] (u51) at (-1.3,1.4) {};
\node[knoten] (u52) at (0,1.4) {};
\node[knoten] (u53) at (1.3,1.4) {};

\node[knoten] (a) at (0,0.5) {};

\draw[very thick, color=blue] (i0) to node [anchor=east]  {} (l21);
\draw[thick, color=myred]  (i0) to node [anchor=west] {} (r21);

\draw[very thick, color=blue] (l31) to (u51);
\draw[thick, color=myred] (r31) to (u51);

\draw[very thick, color=blue] (l32) to (u52);
\draw[thick, color=myred]  (r32) to (u52);

\draw[very thick, color=blue] (l33) to (u53);
\draw[thick, color=myred]  (r33) to (u53);

\draw[very thick, densely dotted, bend right] (u51) to (a);
\draw[very thick, densely dotted] (u52) to (a);
\draw[very thick, densely dotted, bend left] (u53) to (a);

\node at (3,3.7) {$\core$};

\end{tikzpicture}
\end{center}

\caption{
Left: A graph where in any STC-labeling with four strong colors and without weak edges, the edges~$e_1$, $e_2$, $e_3$, and~$e_4$ are part of the same strong color class. Right: A no-instance of Multi-STC with~$c=4$ and~$k=0$, where Rule~\ref{Rule: Kernelrule ECS} does not produce an equivalent instance: The inner rectangles correspond to two copies of the gadget on the left. Observe that all blue edges must have a common strong color, and all red edges must have a common strong color distinct that is not blue. Hence, for any STC-labeling of~$G[\core \cup N(\core)]$ it is not possible to extend the labeling to the dotted edges without violating STC. However, Rule~\ref{Rule: Kernelrule ECS} converts this no-instance into a yes-instance.
} \label{Figure: Counterexample ECS Rule for STC}
\end{figure}

If~$c=1$, the parameter~$\pe_{c-1}=\pe_0$ equals the number~$m$ of edges in~$G$. Hence, \textsc{Multi-STC} admits a trivial~$\pe_{c-1}$-edge kernel in this case. If~$c=2$, any input graph consists of core vertices~$\core$, periphery vertices in~$N(\core)$ and isolated vertices and edges. We can compute an equivalent instance in linear time by deleting these isolated components. \iflong The safeness of this rule is obvious.\fi Afterwards, the graph contains at most~$2 \pe_{c-1}$ core vertices. Since each of these vertices has at most one neighbor outside~$\core$, we have a total number of~$4 \pe_{c-1}$ vertices.

To extend this result to~$c \in \{3,4\}$, we now provide a problem kernel for \textsc{Multi-STC} parameterized by~$(c,\pe_{\lfloor \frac{c}{2} \rfloor +1})$. Let~$(G,c,k)$ be an instance of~\textsc{Multi-STC} with edge-deletion set~$D:=D_{\lfloor \frac{c}{2} \rfloor +1}$, and let~$\mathscr{C}$ and~$\mathscr{P}$ be the core and periphery of~$G$. A subset~$A \subseteq \peri$ is called \emph{periphery component} if it is a connected component in~$G[\peri]$. Furthermore, for a periphery component~$A \subseteq \mathscr{P}$ we define the subset~$A^* \subseteq A$ of \emph{close vertices} in~$A$ as~$A^*:=N(\mathscr{C}) \cap A$, that is, the set of vertices of~$A$ that are adjacent to core vertices. 
The key technique of our kernelization is to move weak edges along paths inside periphery components. 

\begin{definition} \label{Def: Good PC}
Let~$(G,c,k)$ be an instance of~\textsc{Multi-STC} with core vertices~$\core$ and periphery vertices~$\peri$. A periphery component~$A \subseteq \mathscr{P}$ is called \emph{good}, if for every STC-labeling~$L=(S_L^1, \dots, S_L^c, W_L)$ for~$G$ with~$E(A) \subseteq W_L$ there exists an STC-labeling~$L'=(S_{L'}^1, \dots, S_{L'}^c, W_{L'})$ for~$G$ such that
\iflong \begin{enumerate}
\item[1.]\else 1.~\fi $L'|_{E \setminus E(A)}=L|_{E \setminus E(A)}$, and
\iflong \item[2.]\fi $W_{L'} \cap E(A)= \emptyset$.
\iflong \end{enumerate}\fi
\end{definition}

Intuitively, a good periphery component~$A$ is a periphery component where the edges  in~$E(A)$ can always be added to some strong color classes of an STC-labeling, no matter how the other edges of $G$ are labeled. The condition~$E(A) \subseteq W_L$ is a technical condition that makes the proof of the next proposition easier. 

\begin{proposition} \label{Prop: Delete good PCs}
Let~$(G,c,k)$ be an instance of~\textsc{Multi-STC} with core vertices~$\core$ and periphery vertices~$\peri$. Furthermore, let~$A \subseteq \peri$ be a good periphery component. Then,~$(G,c,k)$ is a yes-instance if and only if~$(G- (A \setminus A^*),c,k)$ is a yes-instance.
\end{proposition}
\iflong
\begin{proof}
Let~$\tilde{G}=(\tilde{V},\tilde{E}):=(G- (A \setminus A^*),c,k)$. We show that~$G$ has a~$c$-colored STC-labeling with at most~$k$ weak edges if and only if~$\tilde{G}$ has a~$c$-colored STC-labeling with at most~$k$ weak edges.

Let~$L=(S^1_L, \dots, S^c_L, W_L)$ be a~$c$-colored STC-labeling for~$G$ such that~$|W_L|\leq k$. Then, define by~$\tilde{L}:=(S^1_{{L}}\cap \tilde{E}, \dots, S^c_{{L}}\cap \tilde{E}, W_L \cap \tilde{E})$ a~$c$-colored labeling for~$\tilde{G}$. Obviously~$|W_L \cap \tilde{E}|\leq k$. It remains to show that~$\tilde{L}$ satisfies STC. Since~$\tilde{G}$ is an induced subgraph of~$G$, every two edges~$e_1, e_2 \in \tilde{E}$ forming a~$P_3$, a path on three vertices, in~$\tilde{G}$ also form a~$P_3$ in~$G$. Hence, from the fact that~$L$ satisfies STC we conclude that~$\tilde{L}$ satisfies~STC.

Conversely, let~$\tilde{L}=(S^1_{\tilde{L}}, \dots, S^c_{\tilde{L}}, W_{\tilde{L}})$ be a~$c$-colored STC-labeling for~$\tilde{G}$ such that~$|W_{\tilde{L}}|\leq k$. We define a~$c$-colored labeling~$L:=(S^1_L,\dots,S^c_L,W_L)$ for~$G$ by~$S^j_L := S^j_{\tilde{L}} \setminus  E(A)$ and~$W_L := W_{\tilde{L}} \cup E(A)$. Note that~$L|_{E \setminus E(A)}=\tilde{L}|_{E \setminus E(A)}$ and~$E(A)\subseteq W_L$. Then, by the definition of good periphery components there exists an STC-labeling~$L'=(S_{L'}^1, \dots, S_{L'}^c, W_{L'})$ for~$G$ such that~$L'|_{E \setminus E(A)}=L|_{E \setminus E(A)}=\tilde{L}|_{E \setminus E(A)}$, and $W_{L'} \cap E(A)= \emptyset$. We next show that~$|W_{L'}| \leq k$. Since~$W_{L'} \cap E(A)= \emptyset$ every weak edge of~$L'$ is an element of~$E \setminus E(A)$. From the facts that~$L'|_{E \setminus E(A)}=\tilde{L}|_{E \setminus E(A)}$ and~$|W_{\tilde{L}}| \leq k$ it follows~$|W_{L'}| \leq k$. Therefore, $L'$ is an STC-labeling for~$G$ with at most~$k$ weak edges.
\end{proof}
\fi

In the following, we show that for instances~$(G,c,k)$ with~$c \geq 3$ we can compute an equivalent instance of size~$\mathcal{O}(\pe_{\lfloor \frac{c}{2} \rfloor +1} c)$. 
We first consider all cases where~$c \geq 3$ is odd. In this case, we can prove that all periphery components are good.

\begin{proposition} \label{Prop: Odd c all good}
Let~$(G,c,k)$ be an instance of \textsc{Multi-STC}, where~$c \geq 3$ is odd. Moreover, let~$A \subseteq \mathscr{P}$ be a periphery component. Then, $A$ is good.
\end{proposition}
\iflong
\begin{proof}
Let~$L$ be an arbitrary STC-labeling for~$G$ with~$E(A)\subseteq W_L$. We prove that there is an STC-labeling which is partially equal to~$L$ on~$E\setminus E(A)$ and has no weak edges in~$E(A)$.

Let~$L'$ be an STC-labeling for~$G$ with~$L'|_{E\setminus E(A)} = L|_{E\setminus E(A)}$. If~$W_{L'} \cap E(A) = \emptyset$, nothing more needs to be shown. So, assume there is an edge~$\{u,v\} \in W_{L'} \cap E(A)$. Since~$u,v \in A$, it holds that~$\deg(u)\leq \lfloor \frac{c}{2} \rfloor +1$ and~$\deg(v) \leq \lfloor \frac{c}{2} \rfloor +1$. Then, since~$c$ is odd, the edge~$\{u,v\}$ is incident with at most~$2 \cdot \lfloor \frac{c}{2} \rfloor < c$ edges in~$G$. Consequently, there exists a strong color~$i \in \{1, \dots, c\}$, such that~$\{u,v\}$ can be added to the strong color class~$S^i_{L'}$ and be removed from~$W_{L'}$ without producing any STC violations. This way, we transformed~$L'$ into an STC-labeling~$L''$, such that~$L''|_{E\setminus\{\{u,v\}\}} = L|_{E\setminus\{\{u,v\}\}}$ and~$|W_{L''}| = |W_{L'}| - 1$.  Since~$L$ was arbitrary, the periphery component~$A$ is good by definition.
\end{proof}
\fi
The Propositions~\ref{Prop: Delete good PCs} and~\ref{Prop: Odd c all good} guarantee the safeness of the following rule:

\begin{redrule} \label{Rule: odd c}
If~$c$ is odd, remove~$A \setminus A^*$ from all periphery components~$A \subseteq \mathscr{P}$.
\end{redrule}

\begin{proposition} \label{Prop: Kernel for odd c}
Let~$(G=(V,E),c,k)$ be an instance of~\textsc{Multi-STC} where~$c \geq 3$ is odd. Then, we can compute an instance~$(G'=(V',E'),c,k)$ in~$\mathcal{O}(n+m)$ time such that~$|V'| \leq 2\cdot \pe_{\lfloor \frac{c}{2} \rfloor + 1}\cdot (\lfloor \frac{c}{2} \rfloor +1)$, and~$|E'| \in \mathcal{O}(\pe_{\lfloor \frac{c}{2} \rfloor + 1} \cdot c^2)$.
\end{proposition}
\iflong
\begin{proof}
Let~$\core$ be the set of core vertices of~$G$, and~$\peri$ the set of periphery vertices of~$G$. We compute~$(G'=(V',E'),c,k)$ from~$G$ by applying Rule~\ref{Rule: odd c} exhaustively. This can be done by computing~$G[\mathscr{C}\cup N(\mathscr{C})]$ in~$\mathcal{O}(n+m)$ time. We next analyze the size of~$\mathscr{C}\cup N(\mathscr{C})$.

Since~$|\core| \leq 2 \pe_{\lfloor \frac{c}{2} \rfloor +1}$, and every~$v \in \core$ has at most~$\lfloor \frac{c}{2} \rfloor +1$ neighbors in~$\peri$, there are at most~$2\cdot \pe_{\lfloor \frac{c}{2} \rfloor + 1}\cdot (\lfloor \frac{c}{2} \rfloor +1)$ vertices in~$V'$. Since each vertex is incident with at most ~$\lfloor \frac{c}{2} \rfloor +1$ edges, we conclude~$|E'| \in \mathcal{O}(\pe_{\lfloor \frac{c}{2} \rfloor + 1} \cdot c^2)$.
\end{proof}
\fi

It remains to consider instances where~$c$ is an even number and~$c\ge 4$. In this case, not every periphery component is good (Figure~\ref{Figure: Counterexample ECS Rule for STC} shows an example), so we need to identify good periphery components more carefully. The first rule removes isolated periphery components.

\begin{redrule} \label{Rule: Isolated PC}
Remove periphery components~$A\subseteq \mathscr{P}$ with~$A^*=\emptyset$ from~$G$.
\end{redrule}

\begin{proposition}
Rule~\ref{Rule: Isolated PC} is safe.
\end{proposition}

\iflong
\begin{proof}
We prove that periphery components~$A$ with~$A^*=\emptyset$ are good. Safeness of Rule~\ref{Rule: Isolated PC} then follows by Proposition~\ref{Prop: Delete good PCs}.

Let~$L=(S^1_L, \dots, S^c_L, W_L)$ be an STC-labeling with~$E(A) \subseteq W_L$. Since~$c \geq 4$ every vertex in~$A$ has degree at most~$\lfloor \frac{c}{2}\rfloor +1 \leq c-1$. Thus, there exists an STC-labeling~$L'=(S^1_{L'}, \dots,S^c_{L'},W_{L'})$ for~$G[A]$ with~$W_{L'}=\emptyset$ due to Vizing's Theorem. We define a labeling
\begin{align*}
L'':= (S^1_L \cup S^1_{L'}, \dots, S^c_L \cup S^c_{L'}, W_L \setminus E(A)).
\end{align*}
Since~$E(A,V\setminus A)=\emptyset$, $L''$ is an STC-labeling. Moreover, it holds that~$L''|_{E\setminus E(A)}=L|_{E \setminus E(A)}$ and~$W_{L''} \cap E(A) = \emptyset$. Therefore, $A$ is a good periphery component.
\end{proof}
\fi

%
%
%

The intuition for the next lemma is that the small degree of vertices in periphery components can be used to  `move' weak edges inside periphery components, the key technique of our kernelization. More precisely, if there is an edge-simple path in a periphery component, that starts with a weak edge, we can either move the weak edge to the end of that path by keeping the same number of weak edges or find a labeling with fewer weak edges.

\begin{lemma}\label{Lemma: Move the weak edges}
Let~$A \subseteq \mathscr{P}$, let $L$ be an STC-labeling of~$G$, and let~$e \in W_L \cap E(A)$ be a weak edge in~$E(A)$. Furthermore, let~$P=(v_1,v_2,\dots,v_{r-1},v_r)$ be an edge-simple path in~$G[A]$ with~$\{v_1,v_2\}=e$ and color sequence~$Q_L^P=(q_1=0,q_2,q_3,\dots,q_{r-1})$ under~$L$. Then, there exists an STC-labeling~$L'$ with~$L'|_{E \setminus E(P)}=L|_{E\setminus E(P)}$ such that
\begin{align*}
Q_{L'}^P = (q_2,q_3,\dots, q_{r-1}, 0) \text{ or } |W_{L'}| < |W_L|.
\end{align*}
\end{lemma}

\begin{proof}
We prove the statement by induction over the length~$r$ of~$P$.

\textit{Base Case:} $r=2$. Then,~$P=(v_1,v_2)$ and~$Q_L^P=(0)$. We can trivially define the labeling~$L'$ by setting~$L':=L$.


\textit{Inductive Step:} Let~$P=(v_1, \dots, v_r)$ be an edge-simple path with color sequence~$Q_L^P=(0,q_2,\dots,q_{r-1})$ under~$L$. Consider the edge-simple subpath~$P'=(v_1, \dots, v_{r-1})$. By induction hypothesis there exists an STC-labeling~$L''$ for $G$ with~$L''|_{E \setminus E(P')}=L|_{E\setminus E(P')}$, such that~$Q_{L''}^{P'} = (q_2,q_3,\dots, q_{r-2}, 0) \text{ or } |W_{L''}| < |W_L|$.

\textbf{Case 1:} $|W_{L''}| < |W_L|$. Then, we define~$L'$ by~$L':=L''$. 

\textbf{Case 2:} $|W_{L''}| \geq |W_L|$. Then,~$Q_{L''}^{P'} = (q_2,q_3,\dots, q_{r-2}, 0)$. Since~$Q_{L''}^{P'}$ contains the same elements as~$Q_{L}^{P'}$ and~$L''|_{E \setminus E(P')}=L|_{E\setminus E(P')}$, we have ~$|W_{L''}| = |W_L|$.

\textbf{Case 2.1:} There exists an edge~$e \neq \{v_{r-1},v_r\}$ with~$e \in S^{q_{r-1}}_{L''}$ that is incident with~$ \{v_{r-2},v_{r-1}\}$. From the fact that~$\deg(v_{r-2}) \leq \lfloor \frac{c}{2} \rfloor +1$ and~$\deg(v_{r-1}) \leq \lfloor \frac{c}{2} \rfloor +1$, we conclude that~$\{v_{r-2},v_{r-1}\}$ is incident with at most~$c$ other edges of~$G$. Since two of these incident edges have the same strong color~$q_{r-1}$ under~$L''$, the edge~$\{v_{r-2},v_{r-1}\}$ is incident with at most~$c-1$ edges of distinct strong colors under~$L''$. Consequently, there exists a strong color~$i \in \{1, \dots, c\}$, such that~$\{v_{r-2},v_{r-1}\}$ can safely be added to the strong color class~$S^i_{L''}$ and be removed from~$W_{L''}$ without producing any strong~$P_3$. This way, we transformed~$L''$ into an STC-labeling~$L'$, such that~$L'|_{E \setminus E(P')}=L|_{E\setminus E(P')}$ and~$|W_{L'}| < |W_L|$.

\textbf{Case 2.2:} There is no edge $e \neq \{v_{r-1},v_r\}$ with~$e \in S^{q_{r-1}}_{L''}$ that is incident with~$ \{v_{r-2},v_{r-1}\}$. We then define~$L'$ by
\begin{align*}
W_{L'}&:=W_{L''} \cup \{\{v_{r-1},v_r\}\} \setminus \{\{v_{r-2},v_{r-1}\}\} \text{, and}\\
S^{q_{r-1}}_{L'}&:=S^{q_{r-1}}_{L''} \cup \{\{v_{r-2},v_{r-1}\}\} \setminus \{\{v_{r-1},v_{r}\}\}.
\end{align*}
Note that~$Q_{L'}^P = (q_2,q_3,\dots, q_{r-1}, 0)$ and $L'|_{E \setminus E(P)}=L|_{E\setminus E(P)}$. Moreover, since~$P$ is edge-simple, the edge~$\{v_{r-1},v_r\}$ does not lie on~$P'$ and since~$L''|_{E \setminus E(P')}=L|_{E\setminus E(P')}$, it holds that~$\{v_{r-1},v_r\} \in S^{q_{r-1}}_{L''}$. Therefore, every edge has exactly one color under~$L'$. It remains to show that~$L'$ satisfies STC. Assume towards a contradiction, that this is not the case. Then, since~$L''$ satisfies STC, there exists an induced~$P_3$ on~$\{v_{r-2},v_{r-1}\}\in S^{q_{r-1}}_{L'}$ and some edge~$e \in S^{q_{r-1}}_{L'}$. Since~$\{v_{r-1},v_r\} \in W_{L'}$ and~$L'|_{E \setminus \{\{v_{r-2},v_{r-1}\},\{v_{r-1},v_{r}\}\}}=L''|_{E \setminus \{\{v_{r-2},v_{r-1}\},\{v_{r-1},v_{r}\}\}}$, the edge~$e \neq \{v_{r-1},v_r\}$ is incident with~$\{v_{r-2},v_{r-1}\}$ and it holds that~$e \in S^{q_{r-1}}_{L''}$. This contradicts the condition of Case~2.2.
\end{proof}

We will now use Lemma~\ref{Lemma: Move the weak edges} to show useful properties of periphery components. First, if there are two weak edges in one periphery component~$A$, we can make these two weak edges incident, which then helps us to define a new labeling that has fewer weak edges in~$A$:

\begin{proposition}\label{Prop: At most one weak in PC}
Let~$A \subseteq \mathscr{P}$ be a periphery component and let~$L$ be an STC-labeling for~$G$. Then, there exists an STC-labeling~$L'$ with~$L'|_{E \setminus E(A)} = L|_{E \setminus E(A)}$ and~$|W_{L'} \cap E(A)| \leq~1$.
\end{proposition}
\iflong
\begin{proof}
If $|W_{L} \cap E(A)| \leq 1$ the statement already holds for~$L'=L$. So, assume there are two distinct edges~$e_1, e_2 \in W_{L} \cap E(A)$. In this case, we construct an STC-labeling which is partially equal to~$L$ on~$E\setminus E(A)$ and has strictly fewer weak edges in~$E(A)$ than~$L$, which then proves the claim.

 Since periphery components are connected components in~$G[\mathscr{P}]$, there exists an edge-simple path~$P=(v_1, \dots, v_r)$ in~$G[A]$ such that~$e_1=\{v_1,v_2\}$ and~$e_2=\{v_{r-1},v_r\}$. Applying Lemma~\ref{Lemma: Move the weak edges} on the edge-simple subpath~$P'=(v_1, \dots, v_{r-1})$ gives us an STC-labeling~$L'$ with~$L'|_{E \setminus E(P)} = L|_{E \setminus E(P)}$ such that~$|W_{L'}| < |W_L|$ or~$Q_{L'}^{P'}=(q_2,q_3, \dots, q_{r-2},0)$.

In case of~$|W_{L'}| < |W_L|$, nothing more needs to be shown. So, assume~$|W_{L'}| = |W_L|$. It follows that~$Q_{L'}^{P'}=(q_2,q_3, \dots, q_{r-2},0)$ and therefore~$Q_{L'}^P=(q_2,q_3, \dots, q_{r-2},0,0)$. Then,~$e_1$ and~$e_2$ are weak under~$L'$. 
Since~$\deg(v_{r-1})\leq \lfloor \frac{c}{2} \rfloor+1$ and $\deg(v_{r}) \leq \lfloor \frac{c}{2} \rfloor+1$, the edge~$e_2$ is incident with at most~$c$ edges. Since at least one of these incident edges is weak,~$e_2$ is incident with at most~$c-1$ edges of distinct strong colors. Consequently, there exists a strong color color~$i \in \{1, \dots, c\}$ such that~$e_2$ can be added to the strong color class~$S^i_{L'}$ and deleted from~$W_{L'}$ without violating~STC. This way, we transformed~$L'$ into an STC-labeling~$L''$ such that~$L''|_{E \setminus E(A)} = L|_{E \setminus E(A)}$ and~$|W_{L''} \cap E(A)| < |W_{L} \cap E(A)|$.
\end{proof}
\fi
Next, we use Proposition~\ref{Prop: At most one weak in PC} to identify specific good components.

\begin{proposition}\label{Prop: Less conflict edges}
Let~$A \subseteq \mathscr{P}$ be a periphery component such that there exists an edge~$\{u,v\} \in E(A)$ which forms an induced~$P_3$ with less than~$c$ other edges in~$G$. Then,~$A$~is~good.
\end{proposition}
\iflong
\begin{proof}
Let~$L$ be an arbitrary STC-labeling for~$G$ with~$E(A) \subseteq W_L$. We prove that there is an STC-labeling which is partially equal to~$L$ on~$E\setminus E(A)$ and has no weak edges in~$E(A)$.

Let~$L'$ be an STC-labeling for~$G$ with~$L'|_{E \setminus E(A)} = L|_{E \setminus E(A)}$. If~$W_{L'} \cap E(A) = \emptyset$, nothing more needs to be shown. So, let~$W_{L'} \cap E(A) \neq \emptyset$. By Proposition~\ref{Prop: At most one weak in PC} we can assume that there is one unique edge~$e \in W_{L'} \cap E(A)$. Since~$A$ is a connected component in~$G[\mathscr{P}]$, there exists an edge-simple path~$P=(v_1, \dots, v_r)$ such that~$\{v_1,v_2\}=e$, and~$\{v_{r-1}, v_r\}=\{u,v\}$ with~$Q_{L'}^P=(0,q_2, \dots, q_{r-1})$. By Lemma~\ref{Lemma: Move the weak edges}, there exists an STC-labeling~$L''$ with~$L''|_{E \setminus E(A)} = L|_{E \setminus E(A)}$ such that~$|W_{L''}| < |W_L|$ or~$Q_{L''}^P=(q_2, \dots, q_{r-1},0)$. In case of~$|W_{L''}| < |W_L|$, nothing more needs to be shown. Otherwise, the edge~$e$ is weak under~$L''$. Since~$e$ is part of less than~$c$ induced~$P_3$s in~$G$, there exists one strong color~$i \in \{1, \dots, c\}$, such that~$e$ can safely be added to~$S^i_{L''}$ and be removed from~$W_{L''}$ without violating STC. This way, we transform~$L''$ into an STC-labeling~$L'''$ with~$L'''|_{E \setminus E(A)} = L|_{E \setminus E(A)}$ and~$W_{L'''} \cap E(A) = \emptyset$.

Since~$L$ was arbitrary, the periphery component~$A$ is good by definition.
\end{proof}
\fi

\begin{proposition}\label{Prop: One low degree is good}
Let~$A \subseteq \mathscr{P}$ be a periphery component such that there exists a vertex~$v \in A$ with~$\deg_G(v) < \lfloor \frac{c}{2} \rfloor +1$. Then,~$A$ is good.
\end{proposition}
\iflong
\begin{proof}
If~$|A| = 1$, then $A$ is obviously good, since~$E(A)= \emptyset$. Let~$|A| \geq 2$. Since~$A$ contains at least two vertices and forms a connected component in~$G[\peri]$ there exists a vertex~$u \in A$, such that~$\{u,v\} \in E(A)$. Since~$\deg_G(v) < \lfloor \frac{c}{2} \rfloor +1$, and~$\deg_G(u) \leq \lfloor \frac{c}{2} \rfloor +1$, the edge~$\{u,v\}$ forms induced~$P_3$s with less than~$c$ other edges in~$G$. Then, by Proposition~\ref{Prop: Less conflict edges} we conclude that~$A$ is good.
\end{proof}
\fi
Propositions~\ref{Prop: Delete good PCs} and~\ref{Prop: One low degree is good} guarantee the safeness of the following rule.

\begin{redrule} \label{Rule: delete low degree PCs}
If there is a periphery component~$A \subseteq \mathscr{P}$ with~$A \setminus A^* \neq \emptyset$ such that there exists a vertex~$v \in A$ with~$\deg(v) < \lfloor \frac{c}{2} \rfloor +1$, then delete~$A \setminus A^*$ from~$G$.
\end{redrule}

\begin{proposition}\label{Prop: Triangles are good}
Let~$A \subseteq \mathscr{P}$ be a periphery component such that there exists an edge~$\{u,v\} \in E(A)$ which is part of a triangle~$G[\{u,v,w\}]$ in~$G$. Then,~$A$ is good.
\end{proposition}
\iflong
\begin{proof}
Since~$u,v \in A$, we know~$\deg_G(u) \leq \lfloor \frac{c}{2} \rfloor +1$ and~$\deg_G(v) \leq \lfloor \frac{c}{2} \rfloor +1$. Since~$u,v$ are part of a triangle in~$G$, it follows that~$\{u,v\}$ forms an induced~$P_3$ with less than~$c$ other edges in~$G$. Then, by Proposition~\ref{Prop: Less conflict edges} we conclude that~$A$ is good.
\end{proof}
\fi
Propositions~\ref{Prop: Delete good PCs} and~\ref{Prop: Triangles are good} guarantee the safeness of the following rule.

\begin{redrule} \label{Rule: delete triangle PCs}
If there is a periphery component~$A \subseteq \mathscr{P}$ with~$A \setminus A^* \neq \emptyset$ such that there exists an edge~$\{u,v\} \in A$ which is part of a triangle~$G[\{u,v,w\}]$ in~$G$, then delete~$A \setminus A^*$ from~$G$.
\end{redrule}

For the rest of this section we consider instances~$(G,c,k)$ for \textsc{Multi-STC}, that are reduced regarding Rules~\ref{Rule: Isolated PC}--\ref{Rule: delete triangle PCs}. Observe that these instances only contain triangle-free periphery components~$A$ where every vertex~$v \in A$ has~$\deg(v)= \lfloor \frac{c}{2} \rfloor +1$. Since~\textsc{ECS} and~\textsc{Multi-STC} are the same on triangle-free graphs one might get the impression that we can use Vizing's Theorem to prove that all periphery components in~$G$ are good. Consider the example in Figure~\ref{Figure: Counterexample ECS Rule for STC} to see that this is not necessarily the case. 


We now continue with the description of the kernel for~\textsc{Multi-STC}. Let~$(G,c,k)$ be an instance of \textsc{Multi-STC} that is reduced regarding Rules~\ref{Rule: Isolated PC}--\ref{Rule: delete triangle PCs}. We analyze the periphery components of~$G$ that contain cycles. In this context, a~\emph{cycle} (of length~$r$) is an edge-simple path~$P=(v_0,v_1,\dots,v_{r-1},v_0)$ where the last vertex  and the first vertex of~$P$ are the same, and all other vertices occur at most once in~$P$. We will see that acyclic periphery components---which are periphery components~$A \subseteq \peri$ where~$G[A]$ is a tree---are already bounded in~$c$ and~$\pe_{\lfloor \frac{c}{2} \rfloor + 1}$. To remove the other components, we show that periphery components with cycles are always good. To this end we show two lemmas. The intuitive idea behind Lemmas~\ref{Lemma: Turn around cycles} and~\ref{Lemma: Move strong color in cycles} is, that we use Lemma~\ref{Lemma: Move the weak edges} to rotate weak and strong edge-colors around a cycle.

\begin{lemma} \label{Lemma: Turn around cycles}
Let~$A \subseteq \mathscr{P}$ be a periphery component, and let~$L$ be an STC-labeling for~$G$. Moreover, let~$P=(v_0, v_1, \dots, v_{r-1},v_0)$ be a cycle in~$A$ such that~$W_L \cap E(P) \neq \emptyset$ and let~$Q_L^P=(q_0,q_1,\dots,q_{r-1})$ be the color sequence of~$P$ under~$L$. Then, there exist STC-labelings~$L_0,L_1,L_2,\dots,L_{r-1}$ for~$G$ such that~$L_{i}|_{E \setminus E(P)}=L|_{E \setminus E(P)}$ and 
\iflong
\begin{align*} 
  Q^P_{L_i}(j) = q_{(i+j)\mod r}  \text{ or } |W_{L_i}| < |W_L|
\end{align*}
\fi
for all~$i,j \in \{0,\dots,r-1\}$.
\end{lemma}
\iflong
\begin{proof}
Without loss of generality we assume that~$\{v_0,v_1\}\in W_L$ and therefore~$q_0=0$. We prove the existence of the labelings~$L_i$ with~$i \in \{0,1,\dots,r-1\}$ by induction over~$i$.

\textit{Base Case:} $i=0$. In this case we set~$L_0:=L$.

\textit{Inductive Step:} By inductive hypothesis, there is a labeling~$L_{i-1}$ with~$|W_{L_{i-1}}| < |W_L|$ or
\begin{align*}
Q^P_{L_{i-1}}(j) = q_{(i-1 + j) \mod r}.
\end{align*}
If~$|W_{L_{i-1}}|<|W_L|$, then we define~$L_i$ by~$L_i:=L_{i-1}$ and nothing more needs to be shown. Otherwise, we consider~$P'=(v_{r-i+1},v_{r-i+2}, \dots, v_{r-1},v_0,v_1,\dots,v_{r-i+1})$. Note that~$P'$ describes the same cycle as~$P$ by rotating the vertices. More precisely,
\begin{align*}
P(j) = P'((j+i-1) \mod r).
\end{align*}
Therefore,~$P'$ is edge-simple and has the color sequence~$Q^{P'}_{L_{i-1}}=(q_0=0,q_1,\dots,q_{r-1})$. By Lemma~\ref{Lemma: Move the weak edges}, there exists an STC-labeling~$L_i$ with $L_i|_{E \setminus E(P)}=L_{i-1}|_{E \setminus E(P)}$, such that~$|W_{L_{i}}| < |W_{L_{i-1}}|$ or
\begin{align*}
Q^{P'}_{L_i}(j) = q_{(j+1) \mod r}.
\end{align*}
In case of~$|W_{L_{i}}| < |W_{L_{i-1}}|$, nothing more needs to be shown. Otherwise, observe that
\begin{align*}
Q^P_{L_i}(j) = Q^{P'}_{L_i}((j+i-1) \mod r) = q_{(j+i) \mod r}
\end{align*}
which completes the inductive step.
\end{proof}
\fi
\begin{lemma}\label{Lemma: Move strong color in cycles}
Let~$A \subseteq \mathscr{P}$ be a periphery component, let~$L$ be an STC-labeling. Moreover, let~$P=(v_0, v_1, \dots,v_{r-1},v_0)$ be a cycle in~$A$ with~$W_L \cap E(P) \neq \emptyset$, and let~$e_1,e_2 \in E(P)$ with~$e_2 \in S^q_L$ for some strong color~$q\in \{1, \dots, c\}$. Then, there exists an STC-labeling~$L'$ with~$L'|_{E \setminus E(P)} = L|_{E \setminus E(P)}$ such that~$e_1 \in S^q_{L'}$ or~$|W_{L'}| < |W_L|$.
\end{lemma}
\iflong
\begin{proof}
Let~$Q_L^P:=(q_0,q_1,\dots,q_{r-1})$. Without loss of generality assume that~$\{v_0,v_1\} \in W_L$ and~$e_2 = \{v_t, v_{t+1}\}$ for some~$t \in \{1,\dots,r-1\}$. It then holds, that~$q_0=0$, and~$q=q_t$. Furthermore, since~$e_1 \in E(P)$ we have~$e_1=\{P(j),P(j+1)\}$ for some~$j \in \{0,1, \dots, {r-1}\}$.

Consider the STC-labelings~$L_0, L_1, L_2, \dots L_{r-1}$ from Lemma~\ref{Lemma: Turn around cycles}. If for one such labeling~$L_i$ it holds that~$|W_{L_i}|<|W_L|$, then nothing more needs to be proven. Otherwise, set~$i:=(t-j) \mod r$. We show that~$e_1 \in S^{q_t}_{L_{i}}$ by proving~$Q^P_{L_{i}}(j) =q_t$ as follows:
\begin{align*}
Q^P_{L_i}(j) = q_{(i+j) \mod r} = q_{((t-j) \mod r) + j) \mod r}=q_{(t-j+j) \mod r }= q_t.
\end{align*}
\end{proof}
\fi
We next use Lemma~\ref{Lemma: Move strong color in cycles} to prove that periphery components with cycles are good.

\begin{proposition} \label{Prop: Cycle PCs are good}
Let~$(G=(V,E),c,k)$ be a reduced instance of \textsc{Multi-STC} regarding rules~\ref{Rule: Isolated PC}--\ref{Rule: delete triangle PCs}, where~$c \geq 4$ is even. Let~$A \subseteq \mathscr{P}$ be a periphery component in~$G$ such that~$A \setminus A^*\neq \emptyset$ and there is a cycle~$P=(v_0,v_1,\dots,v_{r-1},v_0)$ in~$G[A]$. Then,~$A$ is good.
\end{proposition}

\iflong
\begin{proof}
Without loss of generality, we assume that the cycle~$P$ has no chords. Otherwise we replace~$P$ by the shorter cycle. Let~$L$ be an arbitrary STC-labeling for~$G$ with~$E(A) \subseteq W_L$. We prove that there is an STC-labeling which is partially equal to~$L$ on~$E\setminus E(A)$ and has no weak edges in~$E(A)$.

Let~$L'$ be an STC-labeling for~$G$ with~$L'|_{E \setminus E(A)}= L|_{E \setminus E(A)}$. If~$W_{L'} \cap E(A) = \emptyset$, nothing more needs to be shown. So, let~$W_{L'} \cap E(A) \neq \emptyset$. Then, by Proposition~\ref{Prop: At most one weak in PC} we can assume that there exists one unique~$e \in W_{L'} \cap E(A)$. Moreover, by Lemma~\ref{Lemma: Move the weak edges} we assume without loss of generality that~$e=\{v_0,v_1\}$. Then,~$P$ is a cycle with~$E(P) \cap W_{L'} \neq \emptyset$ in~$G[A]$. 

We will use Lemma~\ref{Lemma: Move strong color in cycles} to transform~$L'$ into an STC-labeling without weak edges in~$E(A)$. To this end, we need to introduce some notation: For a vertex~$v \in V(P)$, we let~$\out(v):=\{i \in \{1, \dots, c\} \mid \exists e\in E \setminus E(P): e \cap V(P) = \{v\} \land e \in S^i_{L'}\}$ denote the set of strong colors of incident edges of~$v$ that are not in~$E(P)$. Consider the following case distinction.

\textbf{Case 1:} There exists an edge~$\{P(j),P(j+1)\} \in E(P)$ that has a strong color~$q \in \bigcup_{v \in P} \out(v)$ under~$L$. Let~$v \in P$ be a vertex with~$q \in \out(v)$, and let~$e \in E(P)$ with~$v \in e$ be an edge incident with~$v$. Since~$\{P(j),P(j+1)\} \in S^q_{L'}$, Lemma~\ref{Lemma: Move strong color in cycles} guarantees the existence of an STC-labeling~$L''$ with~$L''|_{E\setminus E(P)}=L'|_{E \setminus E(P)}$, such that~$e \in S^q_{L''}$ or~$|W_{L''}| < |W_{L'}|$. 

Assume towards a contradiction that~$e \in S^q_{L''}$. Then, since~$L''|_{E\setminus E(P)}=L'|_{E \setminus E(P)}$ and~$q \in \out(v)$, the vertex~$v$ has two incident edges with the same strong color. Furthermore, since~$G$ is reduced regarding Rule~\ref{Rule: delete triangle PCs}, no edge in~$E(A)$ is part of a triangle. Hence, $v$ is the central vertex of an induced~$P_3$ where both edges have strong color~$q$ under~$L''$. This contradicts the fact that~$L''$ is an STC-labeling. We conclude~$|W_{L''}| < |W_{L'}|$, which implies~$L''|_{E \setminus E(A)} = L|_{E \setminus E(A)}$ and~$E(A) \cap W_{L''} = \emptyset$. Since~$L$ was arbitrary, the periphery component~$A$ is good by definition.

\textbf{Case 2:} There is no edge in~$E(P)$ that has a strong color~$q \in \bigcup_{v \in P} \out(v)$.

\textbf{Case 2.1:} There is a strong color~$q$, such that for some~$j \in \{0, \dots, {r-1}\}$ it holds that~$q \in \out(P(j)) \cap \out(P(j+1))$. In this case, consider the edge-simple subpath~$P'=(v_0,v_1,\dots,v_j,v_{j+1})$. Observe that~$Q^{P'}_{L'}(0)=0$, since~$\{v_0,v_1\} \in W_{L'}$. By Lemma~\ref{Lemma: Move the weak edges}, there exists an STC-labeling~$L''$ with~$L''|_{E \setminus E(P')}=L'|_{E \setminus E(P')}$ such that~$|W_{L''}|<|W_{L'}|$ or~$Q^{P'}_{L''}(j)=0$. In case of~$|W_{L''}|<|W_{L'}|$, nothing more needs to be shown. Otherwise,~$Q^{P'}_{L''}(j)=0$ implies~$\{P(j),P(j+1)\} \in W_{L''}$. Since~$q \in \out(P(j)) \cap \out(P(j+1))$ and~$\deg(P(j))=\deg(P(j+1))= \lfloor \frac{c}{2} \rfloor + 1$, the edge~$\{P(j),P(j+1)\}$ is incident with at most~$c-1$ edges of distinct strong colors under~$L''$. Consequently, we can transform~$L''$ into an STC-labeling~$L'''$ with~$L'''|_{E \setminus E(A)} = L''_{E \setminus E(A)}$ and~$W_{L'''} \cap E(A) = \emptyset$. Hence,~$A$ is a good periphery component.

\textbf{Case 2.2:} For every~$j \in \{0,\dots,{r-1}\}$ it holds that~$\out(P(j)) \cap \out(P(j+1)) = \emptyset$. To handle this case, we need to prove two (in-)equalities, that we state in the following claim.
\begin{claim} \label{Claim: Outsizes}
\iflong It holds that
\begin{enumerate}
\item[a)] $|\out(v)| = \frac{c}{2} -1$ for every~$v \in V(P)$, and
\item[b)] $|\bigcup_{v \in P} \out(v)|\leq c-2$.
\end{enumerate}
\else a) $|\out(v)| \leq \lfloor \frac{c}{2} \rfloor -1$ for every~$v \in V(P)$ \hfill b) $|\bigcup_{v \in P} \out(v)|\leq c-2$.
\fi
\end{claim}
\begin{claimproof}
a) Since~$(G,c,k)$ is reduced regarding Rule~\ref{Rule: delete low degree PCs}, and every vertex~$v \in V(P)$ has exactly two neighbors in~$V(P)$ it holds that~$|\out(v)| \leq \frac{c}{2} -1$ for every~$v \in V(P)$. Since there is no weak edge under~$L'$ in~$E(A) \setminus E(P)$, we conclude~$|\out(v)| = \frac{c}{2} -1$.

b) Since~$L'$ satisfies STC, and~$(G,c,k)$ is reduced regarding Rule~\ref{Rule: delete triangle PCs}, there are at least two different edges in~$E(P)$ that are labeled with distinct strong colors under~$L'$. By the condition that no edge in~$E(P)$ has a strong color in~$\bigcup_{v \in P} \out(v)$, we conclude~$|\bigcup_{v \in P} \out(v)|\leq c-2$. 
\end{claimproof}
Consider the set~$\out(P(j))$ for some~$P(j) \in V(P)$. By Claim~\ref{Claim: Outsizes} a),~$|\out(P(j))| = \frac{c}{2} -1$. Since~$|\bigcup_{v \in P} \out(v)|\leq c-2$ by Claim~\ref{Claim: Outsizes} b), there are exactly $c-2- (\frac{c}{2} -1) = \frac{c}{2} -1$ colors in~$|\bigcup_{v \in P, v \neq P(j)} \out(v)|$. Since~$\out(P(j)) \cap \out(P(j+1)) = \emptyset$ for every~$j \in \{1,\dots,r\}$ we conclude~$\out(P((j-1) \mod r))=\out(P((j+1) \mod r))$. Therefore, there are exactly two disjoint sets~$X_1$ and~$X_2$ such that
\iflong
\begin{align*}
\out(v_i) =
\begin{cases}
X_1 & i \text{ is even}\\
X_2 & i \text{ is odd.}
\end{cases}
\end{align*}
\else
$\out(v_i) = X_1$ for even~$i$ and $\out(v_i) = X_2$ for odd~$i$.
\fi
This also implies that~$P$ is a cycle of even length.

We continue with some intuition for the rest of the proof. We will use Lemma~\ref{Lemma: Move the weak edges} to move exactly one strong color from~$\bigcup_{v \in P} \out(v)$ into~$E(P)$. Since there are two alternating out-sets and the length of the cycle~$P$ is at least four (since~$P$ has even length), we obtain a labeling where one strong color occurs in~$E(P)$ and~$\bigcup_{v \in P} \out(v)$, which we already handled in Case~1. To this end, consider the following claim.

\begin{claim} \label{Claim: One cycle-vertex has nice neighbor}
There is~$v \in V(P)$ such that~$v$ has a neighbor~$w \in N_G(v)$ with~$w \in A \setminus V(P)$.
\end{claim}
\begin{claimproof}
Assume towards a contradiction, that there is no such vertex~$v \in V(P)$. Then, for every~$v\in V(P)$ it holds that~$N_G(v) \subseteq \mathscr{C} \cup V(P)$. Then,~$V(P)$ is a connected component in~$G[\mathscr{P}]$. By definition of periphery components, this implies~$A= V(P)$. Moreover, since~$(G,c,k)$ is reduced regarding Rule~\ref{Rule: delete low degree PCs}, we have~$\deg(v) = \lfloor \frac{c}{2} \rfloor +1$ for every~$v \in V(P)$. From the fact that~$\lfloor \frac{c}{2} \rfloor +1>2$ we conclude~$N_G(v) \cap \mathscr{C} \neq \emptyset$ and therefore, every~$v \in V(P)$ is a close vertex of~$A$. This contradicts the fact that~$A \setminus A^* \neq \emptyset$. 
\end{claimproof}

Now, let~$v_j \in V(P)$ be a vertex such that there exists such~$w \in N_G(v)$ with~$w \in A \setminus V(P)$ as described in Claim~\ref{Claim: One cycle-vertex has nice neighbor}. Note that~$\{v_j,w\} \in S^q_{L'}$ for some strong color~$q$. Also note that there exists a vertex~$v' \in V(P)$ such that~$v'$ is distinct from~$v$ and~$q \in \out(v')$.

Consider the path~$P'=(v_0,v_1, \dots, v_j,w)$ in~$G[A]$. Since~$P$ is edge-simple so is~$P'$. Let~$Q_{L'}^{P'}=(q_0, q_1, \dots, q_{j})$ be the color sequence of~$P'$ under~$L'$. Note that~$q_0=0$, and~$q_j=q$. Then, by Lemma~\ref{Lemma: Move the weak edges} there exists an STC-labeling~$L''$ with~$L''|_{E \setminus E(P)} = L'|_{E \setminus E(P)}$ such that~$|W_{L''}| < |W_{L'}|$ or~$Q_{L''}^{P'}=(q_1, \dots, q_{j-1}, q, 0)$. In case of~$|W_{L''}| < |W_{L'}|$, nothing more needs to be proven. Otherwise, from~$Q_{L''}^{P'}=(q_1, \dots, q_{j-1}, q, 0)$ we conclude~$\{v_{j-1}, v_j\} \in S^q_{L''}$ and~$\{v_j, w\} \in W_{L''}$. Then, consider the path~$P''=(w, v_j, v_{j+1}, \dots, v_r, v_0, v_1)$. Since~$P$ is edge-simple so is~$P''$. Note that~$\{v_{j-1},v_j\} \not \in E(P'')$, and~$Q_{L''}^{P''}(1)=0$. Then, by Lemma~\ref{Lemma: Move the weak edges}, there exists an STC-labeling~$L'''$ with~$L'''|_{E \setminus E(P)} = L''|_{E \setminus E(P)}$ such that~$|W_{L'''}| < |W_{L''}|$ or~$Q_{L'''}^{P''}=(q_j, \dots, q_r, 0)$. In case of~$|W_{L'''}| < |W_{L''}|$ nothing more needs to be shown. Otherwise, from~$L'''|_{E \setminus E(P)} = L''|_{E \setminus E(P)}$ and~$Q_{L'''}^{P''}=(q_j, \dots, q_r, 0)$ it follows that~$\{v_0,v_1\} \in W_{L'''}$, and there exists an edge~$\{v_{j-1}, v_j\} \in S^q_{L'''}$ and a vertex~$v'\in V(P)$ which is incident with an edge~$e' \not \in E(P)$ with~$e \in S^q_{L'''}$. Then,~$L'''$ fulfills the conditions of Case 1.
\end{proof}
\fi

Propositions~\ref{Prop: Cycle PCs are good} and~\ref{Prop: Delete good PCs} imply the safeness of the final rule which together with Rules~\ref{Rule: Isolated PC}--\ref{Rule: delete triangle PCs}
 gives the kernel.
\begin{redrule} \label{Rule: delete cycle PCs}
If there is a periphery component~$A \subseteq \mathscr{P}$ with~$A \setminus A^* \neq \emptyset$ such that there exists a cycle~$P$ in~$G[A]$, then delete~$A \setminus A^*$ from~$G$.
\end{redrule}

\begin{theorem} \label{Theorem: STC xi kernel}
\textsc{Multi-STC} restricted to instances with~$c \geq 3$ admits a problem kernel with~$\Oh(\pe_{\lfloor \frac{c}{2} \rfloor + 1} \cdot c)$ vertices and~$\mathcal{O}(\pe_{\lfloor \frac{c}{2} \rfloor + 1} \cdot c^2)$ edges that can be computed in~$\mathcal{O}(n+m)$ time.
\end{theorem}

\begin{proof}
  \iflong
  Throughout this proof let~$\pe:= 2 \pe_{\lfloor \frac{c}{2} \rfloor +1}$ denote the size of a 2-approximate set~$D_t$. From Proposition~\ref{Prop: Kernel for odd c} we know, that if~$c\geq 3$ is odd, we can compute a problem kernel with at most~$2\cdot \pe\cdot (\lfloor \frac{c}{2} \rfloor +1)$ vertices, and~$\mathcal{O}(\pe \cdot c^2)$ in~$\mathcal{O}(n+m)$ time. Let~$(G=(V,E),c,k)$ be an instance of \textsc{Multi-STC}, where~$c\geq 4$ is an even number. We compute an instance~$(G'=(V',E'),c,k)$ as follows: We start by applying Rules~\ref{Rule: Isolated PC} and \ref{Rule: delete low degree PCs} exhaustively. This can be done in~$\mathcal{O}(n+m)$ time by computing all connected components of~$G[\peri]$ and checking whether they have close vertices or vertices of low degree. Afterwards, we apply the Rules~\ref{Rule: delete triangle PCs} and \ref{Rule: delete cycle PCs} exhaustively. This can also be done in~$\mathcal{O}(n+m)$ time by testing if the connected components in~$G[\peri]$ with non-close vertices contain cycles. Note that at this point it is not important whether a cycle is a triangle or a cycle of length bigger than three.

We next show~$|V'| \leq(c+7) \cdot \pe$. \fi Let~$\mathscr{C}$ be the set of core vertices of~$G'$ and~$\mathscr{P}$ be the set of periphery vertices of~$G'$. Since~$|\core| \leq 2 \pe$, and every~$v \in \core$ is incident with at most~$\frac{c}{2}+1$ edges, there are~$2\pe + 2\pe (\frac{c}{2} + 1) = \pe c + 4 \pe$ vertices in~$\core \cup N(\core)$.
It remains to show that there are at most~$3\pe$ non-close vertices in~$\peri$. 
Consider the \iflong following family of periphery components.
\begin{align*} 
\mathcal{A} := \{ A \subseteq \peri \mid A \text{ is periphery component with } A \setminus A^* \neq \emptyset\} 
\end{align*}
\fi
Since~$G'$ is reduced regarding Rules~\ref{Rule: delete low degree PCs}, \ref{Rule: delete triangle PCs}, and \ref{Rule: delete cycle PCs}, every~$G[A]$ with~$A \in \mathcal{A}$ is a tree, where every vertex~$v \in A$ has degree~$\deg_G(v)=\frac{c}{2}+1$ in~$G$. We define a \emph{leaf vertex} as a vertex~$v \in \bigcup_{A \in \mathcal{A}} A$ with~$\deg_{G[\peri]}(v)=1$. Note that these vertices are exactly the leaves of a tree~$G[A]$ for some~$A \in \mathcal{A}$, and all leaf vertices are close vertices in~$\peri$. Let~$p$ be the number of leaf vertices. We show that~$p \leq 3 \pe$. Since~$(G',c,k)$ is reduced regarding Rule~\ref{Rule: delete low degree PCs}, every vertex~$v \in \bigcup_{A \in \mathcal{A}} A$ has a degree of~$\deg_G(v) = \frac{c}{2}+1$, hence every leaf vertex has exactly~$\frac{c}{2}$ neighbors in~$\core$. We thus have \iflong
\begin{align*}
p \cdot \frac{c}{2} \leq |E(\core, N(\core))|\leq 2\pe( \frac{c}{2}+1),
\end{align*}
\fi
and therefore~$p \leq 2\pe+ \frac{4\pe}{c} \leq 3 \pe$, since~$c \geq 4$. Recall that every non-close vertex~$v$ in some tree~$G[A]$ satisfies~$\deg_{G[A]}(v)= \frac{c}{2}+1>2$. Since a tree has at most as many vertices with degree at least three as it has leaves, we conclude~$|(\bigcup_{A \in \mathcal{A}} A) \setminus (\bigcup_{A \in \mathcal{A}} A^*)| \leq 3\pe$. Hence, there are at most~$3 \pe$ non-close vertices in~$\peri$. Then,~$G'$ contains of at most~$(c+7) \cdot \pe \in \Oh(\pe c)$ vertices, as claimed. Since each vertex is incident with at most~$\frac{c}{2} + 1$ edges,~$G'$ has~$\Oh(\pe c^2)$~edges.
\end{proof}

\iflong
\section{ECS and Multi-STC with Lists}\label{sec:lists}
\newcommand{\meng}[1]{\ensuremath{  \{#1\} }}
\newcommand{\Q}{\ensuremath{ Q }}

In this section we present linear-size kernels for \textsc{EL-Multi-STC} and \textsc{EL-ECS} parameterized by~$\pe_2$ for all~$c$. The formal problem definitions are as follows.
\begin{definition}
Let $G=(V,E)$ be a graph, $\Psi: E \rightarrow 2^{\lbrace 1,2, \dots, c\rbrace}$ a mapping for some value $c \in \mathds{N}$ and $L=(S^1_L, \dots, S^c_L,W_L)$ a $c$-colored STC-labeling. We say that an edge $e \in E$ \emph{satisfies the $\Psi$-list property under $L$} if $e \in W_L$ or $ e \in S^{\alpha}_L$ for some $\alpha \in \Psi(e)$. We call a $c$-colored labeling \emph{$\Psi$-satisfying} if every edge $e \in E$ satisfies the $\Psi$-list property under $L$.
\end{definition}
\begin{quote}
  \textsc{Edge-List Edge-Colorable Subgraph (EL-ECS)}\\
  \textbf{Input}: An undirected graph~$G=(V, E)$,
  integers~$c \in \mathds{N}$ and~$k \in \mathds{N}$ and edge lists $\Psi: E \rightarrow 2^{\lbrace 1,2, \dots, c\rbrace}$.\\
  \textbf{Question}: Is there a $\Psi$-satisfying labeling~$L$ with~$|W_L| \leq k$?
\end{quote}

\begin{quote}
  \textsc{Edge-List Multi Strong Triadic Closure (EL-Multi-STC)}\\
  \textbf{Input}: An undirected graph~$G=(V, E)$,
  integers~$c \in \mathds{N}$ and~$k \in \mathds{N}$ and edge lists $\Psi: E \rightarrow 2^{\lbrace 1,2, \dots, c\rbrace}$.\\
  \textbf{Question}: Is there a $\Psi$-satisfying STC-labeling~$L$ with~$|W_L| \leq k$?
\end{quote}
Before showing the kernelization, we motivate the parameter~$\pe_2$ with the following negative result.

\elhard*
\begin{proof}
We reduce from \textsc{ECS} on triangle-free cubic graphs which is known to be NP-hard \cite{CE91} to \textsc{EL-ECS}.  

Let~$G = (V,E), k = 0, c = 3$ be an instance of \textsc{ECS} such that~$G$ is a triangle-free cubic graph. Let~$k := 0$ and~$\Psi(e) := \meng{1,2,3}$ for all~$e \in E$. Since every edge~$e\in E$ is only allowed to be colored in either~$1,2,$ or~$3$, it is obvious that~$(G$,$c$,$\Psi$,$k)$ is a yes-instance of \textsc{EL-ECS} if and only if~$G$ is three-edge-colorable. Thus \textsc{EL-ECS} is NP-hard for all~$c \geq 3$ on triangle-free cubic graph even if~$\pe_3=k=0$.

Since \textsc{EL-Multi-STC} corresponds to \textsc{EL-ECS} on triangle-free graphs, the reduction to \textsc{EL-Multi-STC} is completely analogous.
\end{proof}

Next, we will present the linear kernel for the parameter $\pe_2$. We first show the kernel for \textsc{EL-ECS} and then show that (after some preprocessing) all rules are also safe for \textsc{EL-Multi-STC}.
Let~$(G,c,\Psi,k)$ be an instance of~\textsc{EL-ECS} with edge-deletion set~$D:=D_{\pe_2}$, and let~$\mathscr{C}$ and~$\mathscr{P}$ be the core and periphery of~$G$ as defined in Section \ref{sec:ed}. Every periphery component~$A \subseteq \mathscr{P}$ is either an isolated cycle, an isolated vertex, or a path. To differentiate between paths that are part of an isolated cycle and those that are not, we give the following definition. 
\begin{definition}   A \textit{bounded-degree path~$($BDP$)$}~$P=(v_1, \dots, v_r), r > 1,$ in~$G$ is a simple inclusion maximal path in~$G$ such that~$P$ is not part of an isolated cycle and~$ \deg(v_i) \leq 2$ for every~$v_i \in V(P)$. If at least one endpoint of~$P$ has degree one, then a BDP is called \textit{open}. Furthermore, we call a path~$P$ a \textit{bounded-degree subpath~$($BDSP$)$} if~$P$ is a (not necessarily proper) subpath of a BDP.
\end{definition}
By our definition, it is possible that endpoints of a BDP are core vertices. Obviously, the set of isolated cycles and the set of BDPs is unique and can be computed in~$O(n + m)$~time by finding all vertices~$v\in V$ with~$1 \leq \deg(v) \leq 2$, and then computing all induced subgraphs of these vertices that are not isolated vertices. Since all following rules only work on isolated cycles and BDSPs, we do not need to know $D$, $\mathscr{C},$ or~$\mathscr{P}$ in advance.  

With Algorithm~\ref{Algo} we will give a polynomial-time algorithm for finding an optimal labeling for isolated cycles and isolated BDPs which we will use to remove all of them. Thus the following rules aim to reduce the size of BDPs that are connected with at least one core vertex.
In the following, let~$ I_G(e)$ denote the incident edges of an edge~$e$ in~$G$. If~$G$ is clear from the context, we just write~$I(e)$. Observe that~$|I(e)|\leq 2$ for every edge~$e$ that is on a BDP. Furthermore, for a given labeling~$L$, an edge set~$E$, and a vertex~$v$, let~$\out^E_L(v) := \meng{L(e) \in \meng{1,\dots, c} \mid e \in E, v \in e}$ denote the set of strong colors of edges~$e \in E$ incident with~$v$ under~$L$. For a better use, we also extend the definition of~$\out^E_L$ to edges by setting~$\out^E_L(\meng{u,v}) := \out^E_L(u) \cup \out^E_L(v)$. Observe that~$|\out^{E \setminus E(P)}_L(e)|$ is at most two for every edge~$e$ that is on a BDSP~$P$ and at most one if~$P$ has also length at least 3.

\begin{redrule}\label{Rule: delete full or empty edges}
If~$(G, \Psi, k)$ contains an edge~$e$ on a BDP $P = (v_1, \dots, v_r)$ such that $P$ has length at least three or is an open BDP, then remove~$e$ if at least one of the following holds:
\begin{enumerate}
\item\label{Case: empty list}~$\Psi(e) = \emptyset$,
\item\label{Case: list > neighbors}~$|\Psi(e)| > |I(e)|$  or
\item\label{Case: conflictfree color}~$\Psi(e) \setminus \left( \bigcup_{e' \in I(e)} \Psi(e') \right) \neq \emptyset, v_1 \not\in e,$ and $v_r \not\in e$

\end{enumerate}
and set~$\Psi := \Psi|_{E \setminus \meng{e}}$. If~$ \Psi(e) = \emptyset$, also decrease~$k$ by one.
\end{redrule}

\begin{proposition}
Rule~\ref{Rule: delete full or empty edges} is safe and can be exhaustively applied in~$\Oh(n^2)$~time.
\end{proposition}
\begin{proof}
Let~$(G, \Psi, k)$ be an \textsc{EL-ECS} instance with an edge~$e$ satisfying the conditions of Rule~\ref{Rule: delete full or empty edges}. Furthermore, let~$(G', \Psi', k')$ be the modified instance of \textsc{EL-ECS} constructed by Rule~\ref{Rule: delete full or empty edges}. First, we show that~$(G, \Psi, k)$ is a yes-instance if and only if~$(G', \Psi', k')$ is a yes-instance.

\textbf{Case~\ref{Case: empty list}:~$ \Psi(e) = \emptyset$.} In this case,~$k' = k - 1$.

$(\Rightarrow)$ Let~$L$ be a proper~$\Psi$-satisfying labeling for~$G$ with at most~$k$ weak edges. We define~$L' := L|_{E\setminus \meng{e}}$. Since~$\Psi(e) = \emptyset$ it is obvious that~$L(e) = 0$. Thus~$L'$ has at most~$k-1=k'$ weak edges. It is obvious that~$L'$ is a proper~$\Psi'$-satisfying labeling for~$G'$ since~$L$ is a proper~$\Psi$-satisfying labeling for~$G$.

$(\Leftarrow)$ Let~$L'$ be a proper~$\Psi'$-satisfying edge labeling for~$G'$ with at most~$k'$ weak edges. We define~$L$ by setting~$L(e) := 0$ and~$L|_{E \setminus \meng{e}} := L'|_{E \setminus \meng{e}}$. Obviously,~$L$ is a proper~$\Psi$-satisfying labeling with at most~$k'+1=k$ weak edges for~$G$. 

\textbf{Cases~\ref{Case: list > neighbors} and~\ref{Case: conflictfree color}.} In these cases,~$k' = k$.

$(\Rightarrow)$ Let~$L$ be a proper~$\Psi$-satisfying labeling for~$G$ with at most~$k$ weak edges. We define~$L' := L|_{E\setminus \meng{e}}$. It is obvious that~$L'$ is a proper~$\Psi'$-satisfying labeling for~$G'$ with at most~$k=k'$ weak edges.

$(\Leftarrow)$ Let~$L'$ be a proper~$\Psi'$-satisfying labeling for~$G'$ with at most~$k'$ weak edges and let~$e := \meng{v_1,v_2}$. We choose an arbitrary color~$c_x \in \Psi(e) \setminus (\out^{E'}_{L'}(e))$. This set is non-empty because of the conditions of Cases~\ref{Case: list > neighbors} and~\ref{Case: conflictfree color}. We define~$L$ by setting~$L(e) := c_x$ and~$L|_{E \setminus \meng{e}} := L'|_{E \setminus \meng{e}}$. Since~$c_x \in \Psi(e)$ is a strong color such that no incident edge of~$e$ is colored in~$c_x$ under~$L$, it follows directly that~$L$ is a proper~$\Psi$-satisfying labeling with at most~$k'=k$ weak edges.

Next, we show the running time. Finding an edge $e$ that satisfies the condition of Case~\ref{Case: conflictfree color} can only be done in $\Oh(n)$~time if the sizes of the allowed lists of colors of $e$ and its incident edges is constant. So we assume that Case~\ref{Case: empty list} and Case~\ref{Case: list > neighbors} are exhaustively applied, before we apply Case~\ref{Case: conflictfree color}. Every application of Rule~\ref{Rule: delete full or empty edges} removes one edge from~$G$ and an edge can be found in~$\Oh(n)$~time with the previous argumentation. Consequently, Rule~\ref{Rule: delete full or empty edges} can be exhaustively applied in~$\Oh(n^2)$~time.
\end{proof}

From this point onwards we assume that Rule~\ref{Rule: delete full or empty edges} is exhaustively applied.
Let~$P$ be a BDP of length at least 3 or an open BDP, then for every edge~$e\in E(P)$ it now holds that~$|\Psi(e)| \leq 2$ since~$|I(e)| \leq 2$ and Case~\ref{Case: list > neighbors} is exhaustively applied. In combination with Case~\ref{Case: empty list} it is obvious that~$1 \leq |\Psi(e)| \leq 2$. Furthermore, we get that there is no BDSP~$P'=(v_1,v_2,v_3,v_4)$ where~$e_i := \meng{v_i, v_{i+1}}$ for all~$i \in \meng{1,2,3}$ such that~$\Psi(e_1) = \Psi(e_3), \Psi(e_1) \neq \Psi(e_2),$ and~$|\Psi(e_2)| = 2$ because this implies that~$\Psi(e_2) \setminus \left( \bigcup_{e' \in I(e_2)} \Psi(e') \right) \neq \emptyset$ which is not possible after Case~\ref{Case: conflictfree color} is exhaustively applied. 

The next rule splits the center vertex $v_2$ in a BDSP of length three into two new vertices with one incident edge each, if the allowed colors of both incident edges of $v_2$ are disjoint. Hence, this rule splits some BDP into two open BDPs.

\begin{redrule}\label{Rule: split disjoint neighbors}
If~$(G, \Psi, k)$ contains a BDSP~$P = (v_1, v_2, v_3)$ with~$\Psi(e_1) \cap \Psi(e_2) = \emptyset$ where~$e_1 := \meng{v_1, v_2}, e_2 := \meng{v_2, v_3}$, then remove~$v_2$ from~$G$, add two new vertices~$u, w$, two new edges~$e'_1 := \meng{v_1, u}$ and~$e'_2 := \meng{w, v_3}$ to~$G$ and set~$\Psi(e'_1) := \Psi(e_1)$ and~$\Psi(e'_2) := \Psi(e_2)$.
\end{redrule}

\begin{proposition}
Rule~\ref{Rule: split disjoint neighbors} is safe and can be exhaustively applied in~$\Oh(n^2)$~time.
\end{proposition}
\begin{proof}
Let~$(G, \Psi, k)$ be an \textsc{EL-ECS} instance with~$P = (v_1, v_2, v_3)$ satisfying the conditions of Rule~\ref{Rule: split disjoint neighbors}. Furthermore, let~$(G', \Psi', k)$ be the modified instance of \textsc{EL-ECS} constructed by Rule~\ref{Rule: split disjoint neighbors}. First, we show that~$(G, \Psi, k)$ is a yes-instance if and only if~$(G', \Psi', k)$ is a yes-instance.

$(\Rightarrow)$ Let~$L$ be a proper~$\Psi$-satisfying labeling for~$G$ with at most~$k$ weak edges. We define~$L'(e'_1) := L(e_1), L'(e'_2) := L(e_2), L'|_{E \setminus E(P)} := L|_{E \setminus E(P)}$. First, we show that~$L'$ is a proper~$\Psi'$-satisfying labeling with at most~$k$ weak edges. From the definition of~$L'$ and~$\Psi'$ it is obvious that~$L'$ is~$\Psi'$-satisfying, since~$L$ is~$\Psi$-satisfying and~$\Psi'(e'_1) = \Psi(e_1)$ and~$\Psi'(e'_2) = \Psi(e_2)$. It is also clear that~$L'$ has at most~$k$ weak edges because~$L$ has at most~$k$ weak edges. It remains to show that~$L'$ is a proper labeling. Since~$L$ is a proper labeling and~$I_{G'}(e'_i) \subseteq I_G(e_i)$ for all~$i \in \meng{1,2}$ this condition also holds.

$(\Leftarrow)$ Let~$L'$ be a proper~$\Psi'$-satisfying labeling for~$G'$ with at most~$k$ weak edges. We define $L$ by setting~$L(e_1) := L'(e'_1), L(e_2) := L'(e'_2)$, and $L|_{E \setminus E(P)} := L'|_{E \setminus E(P)}$.
First, we show that~$L$ is a proper~$\Psi$-satisfying labeling with at most~$k$ weak edges. Since~$L'$ is a proper labeling and~$I_G(e_i) = I_{G'}(e'_i) \cup E(P) \setminus \meng{e_i}$ for all~$i \in \meng{1,2}$ it is clear that the only conflict of two incident edges receiving the same strong color under~$L$ could be the edges of~$E(P)$. But since~$\Psi(e_1) \cap \Psi(e_2) = \emptyset$ it follows that~$e'_1$ and~$e'_2$ receive different strong colors under~$L'$ and hence also under~$L$. Thus it is is clear that~$L$ is a proper labeling. Clearly,~$L$ is also~$\Psi$-satisfying and has at most~$k$ weak edges. 

Next, we show the running time. Since every edge~$e$ in a BDP has at most two incident edges, it can only be on at most two BDSPs that satisfy the condition of Rule~\ref{Rule: split disjoint neighbors}. So Rule~\ref{Rule: split disjoint neighbors} can be applied at most~$\Oh(n)$~times and a BDSP~$P$ that fulfills the conditions of Rule~\ref{Rule: delete full or empty edges} can be found in~$\Oh(n)$~time. Consequently, Rule~\ref{Rule: split disjoint neighbors} can be exhaustively applied in~$\Oh(n^2)$.
\end{proof}
From this point onwards we assume that Rule~\ref{Rule: split disjoint neighbors} is exhaustively applied.

The next rule moves all edges with a list of size one in a BDP $P$  to one side of the BDP and all other edges to the other side.  
To prevent this rule to be applied infinitely often, we define an order on the vertices of every BDP. Let~$P=(v_1,\dots,v_r), r \geq 2,$ be a BDP and assume without loss of generality that~$\deg(v_1) \geq \deg(v_r)$, then we define an order~$\prec_P : V(P) \times V(P)$ in a way that~$v_i \prec_P v_j :\Leftrightarrow i < j$.

\begin{redrule}\label{sizeOneHugs}
If~$(G, \Psi, k)$ contains a BDP~$P$ and a subpath~$P' = (v_1, v_2, v_3)$ of~$P$ with~$v_1 \prec_P v_2, |\Psi(e_1)| = 2$ and~$|\Psi(e_2)| = 1$ where~$e_1 := \meng{v_1, v_2}, e_2 := \meng{v_2, v_3}$, then set~$\Psi(e_1) := \Psi(e_1) \setminus \Psi(e_2)$ and~$\Psi(e_2) := \Psi(e_1)$.
\end{redrule}

\begin{proposition}
Rule~\ref{sizeOneHugs} is safe and can be exhaustively applied in~$\Oh(n^2)$~time.
\end{proposition}
\begin{proof}
Let~$(G, \Psi, k)$ be an \textsc{EL-ECS} instance with~$P' = (v_1, v_2, v_3)$ satisfying the conditions of Rule~\ref{sizeOneHugs}. Furthermore, let~$(G, \Psi', k)$ be the modified instance of \textsc{EL-ECS} constructed by Rule~\ref{sizeOneHugs}. First, we show that~$(G, \Psi, k)$ is a yes-instance if and only if~$(G, \Psi', k)$ is a yes-instance. 

$(\Rightarrow)$ Let~$L$ be a proper~$\Psi$-satisfying labeling for~$G$ with at most~$k$ weak edges. Since Rule~\ref{Rule: split disjoint neighbors} is applied exhaustively, we can assume without loss of generality that~$\Psi(e_1) = \meng{1,2}, \Psi(e_2) = \meng{1}$ and thus~$\Psi'(e_1) = \meng{2}$ and~$ \Psi'(e_2) = \meng{1,2}$. Because~$L$ is a proper labeling,~$e_1$ and~$e_2$ are not colored in the same strong color. Hence,~$\Q_L^{P'} \in (\meng{w,1,2} \times \meng{w,1}) \setminus \meng{(1,1)}$.

\textbf{Case 1:~$\Q_L^{P'} \neq (1,0)$}: Obviously,~$L$ is already a~$\Psi'$-satisfying labeling for~$G$ with at most~$k$ weak edges.

\textbf{Case 2:~$\Q_L^{P'} = (1,0)$}: Let~$c_x$ be an arbitrary color of~$\Psi'(e_2) \setminus \out^{E \setminus E(P')}_L(e_2)$. We define~$L'(e_1) := 0$ and~$L'(e_2) := c_x$ and~$L'|_{E \setminus E(P')} := L|_{E \setminus E(P')}$. Obviously,~$L'$ is a proper~$\Psi'$-satisfying labeling for~$G$ with at most~$k-1+1=k$ weak edges.

$(\Leftarrow)$ Since we do not use the order~$\prec_P$, this direction is completely analogous to the first one.

Next, we show the running time. Rule~\ref{sizeOneHugs} can be interpreted as a ''swap'' of the~$\Psi$-values of two incident edges, so we can exhaustively apply Rule~\ref{sizeOneHugs} on every BDP~$P$ in time~$\Oh(|V(P)|^2)$ with a modified version of Bubblesort. Since there are at most~$|V(P)|$ edges that lie on BDPs in~$G$ and thus~$\sum_{P : P \textrm{ is BDP}} |V(P)|^2 \in \Oh(|V|^2)$, Rule~\ref{sizeOneHugs} can be exhaustively applied in~$\Oh(n^2)$~time.
\end{proof}

After Rule~\ref{sizeOneHugs} is exhaustively applied, every BDP~$P = (v_1, \dots , v_r), r \geq 2,$ starts with edges that have only one allowed color and since Rule~\ref{Rule: split disjoint neighbors} is exhaustively applied, this unique color is the same for all these edges. From a specific vertex~$v_t$ onwards all edges have an allowed set of two colors. So the following rules aim to reduce the length of these two subpaths of~$P$ to a constant size. Observe that if~$P$ is open, then it holds that~$\deg(v_r) = 1$ and so~$|I(e_{r-1})| \leq 1$ where~$e_{r-1} := \meng{v_{r-1}, v_r}$. Thus~$|\Psi(e_{r-1})| = 1$ since otherwise Rule~\ref{Rule: delete full or empty edges} is applies. Consequently, after Rule~\ref{sizeOneHugs} is exhaustively applied, for every open BDP~$P$ it holds that~$|\Psi(e)| = 1$ for all~$e\in E(P)$. So the next rule reduces each open BDPs to one of length at most one, so that we only have to handle non-open BDPs afterwards. 

\begin{redrule}\label{Rule: size one at open end}
If~$(G, \Psi, k)$ contains a BDSP~$P = (v_1, v_2, v_3)$ with~$\Psi(e_1) = \Psi(e_2), |\Psi(e_1)| = 1$ and~$\deg(v_3) = 1$ where~$e_1 := \meng{v_1, v_2}, e_2 := \meng{v_2, v_3}$, then remove~$v_2, v_3$ from~$G$ and decrease~$k$ by one. Furthermore, set~$\Psi := \Psi|_{E \setminus E(P)}$
\end{redrule}

\begin{proposition}
Rule~\ref{Rule: size one at open end} is safe and can be exhaustively applied in~$\Oh(n^2)$~time.
\end{proposition}
\begin{proof}
Let~$(G, \Psi, k)$ be an \textsc{EL-ECS} instance with~$P = (v_1, v_2, v_3)$ satisfying the conditions of Rule~\ref{Rule: size one at open end}. Furthermore, let~$(G', \Psi', k-1)$ be the modified instance of \textsc{EL-ECS} constructed by Rule~\ref{Rule: size one at open end}. First, we show that~$(G, \Psi, k)$ is a yes-instance if and only if~$(G', \Psi', k-1)$ is a yes-instance.

$(\Rightarrow)$ Let~$L$ be a proper~$\Psi$-satisfying labeling for~$G$ with at most~$k$ weak edges. Since~$\Psi(e_1) = \Psi(e_2)$, and~$|\Psi(e_1)| = 1$ at least one of these two edges is weak under~$L$ and thus it is obvious that~$L' := L|_{E \setminus E(P)}$ is a~$\Psi'$-satisfying labeling for~$G'$ with at most~$k - 1$ weak edges.

$(\Leftarrow)$ Let~$L'$ be a proper~$\Psi'$-satisfying labeling for~$G'$ with at most~$k-1$ weak edges and let~$c_x$ be the unique color of~$\Psi(e_1)$. We define~$L$ with~$L(e_1) := 0, L(e_2) := c_x$ and~$L|_{E \setminus E(P)} := L'|_{E \setminus E(P)}$. Since~$\deg(v_3) = 1$ it follows that~$I(e_2) = \meng{e_1}$ and thus that~$L$ is a proper~$\Psi$-satisfying labeling for~$G$ with at most~$k-1+1 = k$ weak edges.

Next, we show the running time. Every application of Rule~\ref{Rule: size one at open end} removes two edges from~$G$ and a BDSP~$P$ that fulfills the conditions of Rule~\ref{Rule: size one at open end} can be found in~$\Oh(n)$~time. Consequently, Rule~\ref{Rule: size one at open end} can be exhaustively applied in~$\Oh(n^2)$~time.
\end{proof}

As mentioned earlier, for every open BDP~$P$ it holds that~$\Psi(e_1) = \Psi(e_2)$ and~$|\Psi(e_1)| = 1$ for all~$e_1, e_2 \in E(P)$. So obviously, after Rule~\ref{Rule: size one at open end} is exhaustively applied,~$P$~has length at most two. Furthermore if~$P$ is an isolated BDP it follows with Case~\ref{Case: list > neighbors} of~Rule~\ref{Rule: delete full or empty edges}  that~$E(P) = \emptyset$. Thus after the Rules~\ref{Rule: delete full or empty edges} -~\ref{Rule: size one at open end} are exhaustively applied, there is no edge that lies on an isolated BDP.

\begin{proposition}\label{Prop: solve isolated paths}
The number of weak edges in an optimal~$\Psi$-satisfying labeling for an isolated BDP~$P$ can be computed in~$\Oh(|V(P)|^2)$~time.
\end{proposition}
\begin{proof}
Let~$P$ be an isolated BDP in~$G$ and $\Psi : E \rightarrow 2^{\meng{1,\dots, c}}$, we set~$G_2 := G[V(P)]$. Then~$P$ is obviously an isolated (and thus open) BDP in~$G_2$. We construct the \textsc{EL-ECS} instance~$I := (G_2, \Psi, |E|)$. Since~$|E|$ is a trivial upper bound for the number of weak edges in an optimal labeling for~$G_2$,~$I$ is a yes-instance. Let~$I' := (G', \Psi', k')$ be the reduced instance after we apply Rules~\ref{Rule: delete full or empty edges} -~\ref{Rule: size one at open end} exhaustively. Since~$|E(G')| =\emptyset$ and~$I$ and~$I'$ are equivalent instances it follows that~$|E| - k'$ is the minimum number of weak edges for every optimal~$\Psi$-satisfying labeling of~$G$.
Since Rules~\ref{Rule: delete full or empty edges}--\ref{Rule: size one at open end} can all be exhaustively applied in~$\Oh(|V(P)|^2)$ time, this algorithm also runs in~$\Oh(|V(P)|^2)$~time.
\end{proof}

With this proposition at hand we can also compute the optimal number of weak edges for an for an isolated cycle~$C$.

\begin{algorithm} [t]
\caption{EL-ECS minimum weak edges in cycle}\label{Algo}
\begin{algorithmic}[1]
\State \textbf{Input:} A cycle~$C=(v_1,v_2,\dots,v_r, v_1), r\geq 3$,~$G=(V(C),E(C))$ and~$\Psi : E \rightarrow 2^{\meng{1,\dots,c}}$
\State \textbf{Output:} The number of weak edges in an optimal~$\Psi$-satisfying labeling for~$G$
\State$e_x := \meng{v_1, v_2}$
\State$G' := (V, E \setminus \meng{e_x})$
\State$\Psi' := \Psi|_{E(G') \meng{e_x}}$
\If {$|\Psi(e_x)| \geq 3$}
\State \Return minimum number of weak edges for~$(G', \Psi')$
\Else
\State$k' := 1$ + minimum number of weak edges for~$(G', \Psi')$ 
\For{$\alpha \in \Psi(e_x)$} \label{line:foreach}
\State Set $\Psi'|_{I(e_x)}(e) = \Psi(e) \setminus \meng{\alpha}$
\State$k' := \min(k',$ minimum number of weak edges for~$(G', \Psi'))$
\EndFor
\State \Return$k'$
\EndIf
\end{algorithmic}
\end{algorithm}

\begin{proposition}\label{Prop: solve isolated cycle}
Algorithm~\ref{Algo} is correct and runs in~$\Oh(|V(C)|^2)$~time.
\end{proposition}
\begin{proof}
First, we show the correctness. Let $(C,G,\Psi)$ be the input of the algorithm, then $e_x$ is an arbitrary edge of this cycle. Obviously, after removing $e_x$ from $G$, the remaining graph consists of an isolated BDP for which we can find the number of weak edges in an optimal~$\Psi$-satisfying labeling with Proposition \ref{Prop: solve isolated paths}. 
If $r\geq 3$, it is safe to remove $e_x$ from $G$ with the same argumentation as in Rule \ref{Rule: delete full or empty edges}. Otherwise, we can branch over all possible colors $\alpha \in \Psi(e_x) \cup \meng{0}$ by removing $\alpha$ from the allowed set of strong colors of the incident edges of $e_x$. This is correct, since for every proper $\Psi$-satisfying labeling $L$ for $G$ it holds that $L(e_x) \in \Psi(e_x) \cup \meng{0}$ and $L(e_x) \not\in \out^{E\setminus E(C)}_L(e_x)$. Hence, one of the prelabelings is part of an optimal one.

Next, we show the running time. Since every operation in this algorithm has running time at most $\Oh(|V(C)|^2)$ and Line \ref{line:foreach} is finished after at most two turns, the whole algorithm obviously runs in $\Oh(|V(C)|^2)$~time.
\end{proof}

With Proposition~\ref{Prop: solve isolated cycle} it follows directly that the following Rule is safe and can be be exhaustively applied in~$\Oh(n^2)$~time.

\begin{redrule}\label{Rule: cycles}
  If~$(G, \Psi, k)$ contains an isolated cycle~$C$, then compute the number of weak edges~$k'$ in an optimal~$\Psi|_{E(C)}$-satisfying labeling for~$G[C]$ with Algorithm \ref{Algo}. Remove~$C$ from~$G$ and reduce~$k$ by~$k'$.
\end{redrule}
From this point onwards we assume that Rule~\ref{Rule: cycles} is exhaustively applied. So every periphery component is either an isolated vertex, an open BDSP of length two that is connected with a core vertex or a non-open BDSP. So the following rules aim to reduce every non-open BDP to length at most four.

\begin{redrule}\label{Rule: same lists max length two}
If~$(G, \Psi, k)$ contains a proper BDSP~$P = (v_1, v_2, v_3, v_4)$ with~$\Psi(e_1) = \Psi(e_2) = \Psi(e_3)$ where~$e_i := \meng{v_i, v_{i+1}}$ for all~$i \in \meng{1,2,3}$, then remove~$v_2, v_3$ from~$G$, add a new edge~$e' := \meng{v_1, v_4}$ and set~$\Psi(e) := \Psi(e_1)$. Also decrease~$k$ by one if~$|\Psi(e_1)| = 1$.
\end{redrule}

In other words: if there is a proper BDSP of length four where all three edges have the exact same list of allowed colors, then remove all these edges and connect both endpoints directly by an edge that has the same list of allowed colors as the removed edges.

\begin{proposition}
Rule~\ref{Rule: same lists max length two} is safe and can be exhaustively applied in~$\Oh(n^2)$~time.
\end{proposition}
\begin{proof}
Let~$(G, \Psi, k)$ be an \textsc{EL-ECS} instance with~$P = (v_1, v_2, v_3, v_4)$ satisfying the conditions of Rule~\ref{Rule: same lists max length two}. Furthermore, let~$(G', \Psi', k')$ be the modified instance of \textsc{EL-ECS} constructed by Rule~\ref{Rule: same lists max length two} and let~$P' := (v_1,v_4)$. First, we show that~$(G, \Psi, k)$ is a yes-instance if and only if~$(G', \Psi', k')$ is a yes-instance.  

\textbf{Case 1:~$ |\Psi(e_1)| = 1$.} We can assume without loss of generality that~$\Psi(e_1) = \Psi(e_2) = \Psi(e_3) = \meng{1}$ and thus~$\Psi'(e') = \meng{1}$. 

$(\Rightarrow)$ Let~$L$ be a proper~$\Psi$-satisfying labeling with at most~$k$ weak edges for~$G$. Obviously,~$\Q^P_L$ contains at least one weak color since~$(1,1,1)$ is not a proper labeling for~$E(P)$. Initialize $L'$ with~$L'|_{E \setminus E(P)} = L|_{E \setminus E(P)}$.
First, assume that~$\Q^P_L$ contains exactly one weak color. Then,~$\Q^P_L = (1,0,1)$ since otherwise~$L$ is not a proper labeling. We define~$L'(e') = 1$. Since~$1 \not \in \out^{E \setminus E(P)}_L(e_1) \cup \out^{E \setminus E(P)}_L(e_3)$ it is obvious that~$L'$ is a proper~$\Psi'$-satisfying labeling for~$G'$ with at most~$k-1=k'$ weak edges.
Second, assume that~$\Q^P_L$ contains at least two weak colors. We define~$L'(e') = 0$. Obviously,~$L'$ is a proper~$\Psi'$-satisfying labeling for~$G'$ with at most~$k-2+1=k'$ weak edges.

$(\Leftarrow)$ Let~$L'$ be a proper~$\Psi'$-satisfying labeling with at most~$k'$ weak edges for~$G'$. Initialize $L$ with~$L|_{E \setminus E(P)} = L'|_{E \setminus E(P)}$.
First, assume that~$L'(e') = 0$. We define~$L(e_1) = L(e_3) = 0$ and $L(e_2) = 1$. It is obvious that~$L$ is a proper~$\Psi$-satisfying labeling for~$G$ with at most~$k'-1+2=k$ weak edges.
Second, assume that~$L'(e') = 1$. Since~$1 \not \in \out^{E' \setminus E(P')}_{L'}(e')$, we define $L(e_1) = L(e_3) = 1$ and $L(e_2) = 0$. It is obvious that~$L$ is a proper~$\Psi$-satisfying labeling for~$G$ with at most~$k'+1=k$ weak edges.

\textbf{Case 2:~$ |\Psi(e_1)| = 2$.} We can assume without loss of generality that~$\Psi(e_1) = \Psi(e_2) = \Psi(e_3) = \meng{1,2}$ and thus~$\Psi'(e') = \meng{1,2}$.

$(\Rightarrow)$ Let~$L$ be a proper~$\Psi$-satisfying labeling with at most~$k$ weak edges for~$G$. Initialize $L'$ with~$L'|_{E \setminus E(P)} = L|_{E \setminus E(P)}$. 
First, assume that~$\Q^P_L$ contains no weak color. Hence,~$\Q^P_L \in \meng{(1,2,1),(2,1,2)}$. We assume without loss of generality that~$\Q^P_L = (1,2,1)$ and define~$L'(e') = 1$. Since~$1 \not \in \out^{E \setminus E(P)}_L(e_1) \cup \out^{E \setminus E(P)}_L(e_3)$ it is obvious that~$L'$ is a proper~$\Psi'$-satisfying labeling for~$G'$ with at most~$k = k'$ weak edges.
Second, assume~$\Q^P_L$ contains at least one weak color. We define~$L'(e') = 0$. Obviously,~$L'$ is a proper~$\Psi'$-satisfying labeling for~$G'$ with at most~$k-1+1=k'$ weak edges.

$(\Leftarrow)$ Let~$L'$ be a proper~$\Psi'$-satisfying labeling with at most~$k'=k$ weak edges for~$G'$. Initialize $L$ with~$L|_{E \setminus E(P)} = L'|_{E \setminus E(P)}$.
First, assume that~$L'(e') = 0$. Choose an arbitrary color~$c_x \in \Psi(e_3) \setminus \out^{E \setminus E(P)}_{L'}(e_3)$ and let~$c_y$ be the unique remaining color in~$\Psi(e_2) \setminus \meng{c_x}$. We define~$L(e_1) = 0, L(e_2) = c_y$ and $L(e_3) = c_x$. Obviously,~$L$ is a proper~$\Psi$-satisfying edge with at most~$k'-1+1=k$ weak edges for~$G$.
Second, assume that~$L'(e') \neq 0$. Assume without loss of generality that~$L'(e') = 1$. We define $L(e_1) = L(e_3) = 1$ and $L(e_2) = 2$. Since~$1 \not \in \out^{E' \setminus E(P')}_{L'}(e')$ it is obvious that~$L$ is a proper~$\Psi$-satisfying labeling with at most~$k'=k$ weak edges for~$G$.

Next, we show the running time. Every application of Rule~\ref{Rule: same lists max length two} removes two edge from~$G$ and a BDSP~$P$ that fulfills the conditions of Rule~\ref{Rule: same lists max length two} can be found in~$\Oh(n)$~time. Consequently, Rule~\ref{Rule: same lists max length two} can be exhaustively applied in~$\Oh(n^2)$~time.
\end{proof}

So after Rule~\ref{Rule: same lists max length two} is applied exhaustively, every BDP of length at least five contains at most two edges that have a list of size one. Since we aim to reduce all BDPs to length at most four, the following Rule decreases the number of edges on BDPs that have a list of allowed colors of size two by changing the lists on those edges or removing them. 

\begin{redrule}\label{Rule: P4 with different size two lists}
If~$(G, \Psi, k)$ contains a proper BDSP~$P = (v_1, v_2, v_3, v_4)$ with~$|\Psi(e_i)| = 2, \Psi(e_1) \neq \Psi(e_3)$ and~$\Psi(e_2) \neq \Psi(e_3)$ where~$e_i := \meng{v_i, v_{i+1}}$ for all~$i \in \meng{1,2,3}$, then increase~$k$ by one, and do the following:
\begin{itemize}
\item\label{Case: first two same}\textbf{If~$ \Psi(e_1) = \Psi(e_2)$}, then let~$c_x$ be the unique color of~$\Psi(e_2) \cap \Psi(e_3)$ and set~$\Psi(e_1) = \Psi(e_2) = \meng{c_x}$. 
\item\label{Case: actually four colors}\textbf{If~$ \Psi(e_1) \neq \Psi(e_2)$ and $\Psi(e_1) \cap \Psi(e_3) = \emptyset$}, then let~$c_x$ be the unique color of~$\Psi(e_1) \cap \Psi(e_2)$ and~$c_y$ be the unique color of~$\Psi(e_3) \setminus \Psi(e_2)$. Set~$\Psi(e_1) = \Psi(e_2) = \meng{c_x}$ and~$\Psi(e_3) = \meng{c_x, c_y}$.
\item\label{Case: loop color}\textbf{If~$ \Psi(e_1) \neq \Psi(e_2)$ and $\Psi(e_1) \cap \Psi(e_3) \neq \emptyset$}, then let~$c_x$ be the unique color of~$\Psi(e_1) \cap \Psi(e_3)$, remove~$v_3$ from~$G$ and add a new edge~$e' := \meng{v_2, v_4}$. Furthermore, set~$\Psi(e_1) = \Psi(e') = \meng{c_x}$.
\end{itemize}
\end{redrule}

\begin{proposition}
Rule~\ref{Rule: P4 with different size two lists} is safe and can be exhaustively applied in~$\Oh(n^2)$~time.
\end{proposition}
\begin{proof}
Let~$(G, \Psi, k)$ be an \textsc{EL-ECS} instance with~$P = (v_1, v_2, v_3, v_4)$ satisfying the conditions of Rule~\ref{Rule: P4 with different size two lists}. Furthermore, let~$(G', \Psi', k')$ be the modified instance of \textsc{EL-ECS} constructed by Rule~\ref{Rule: P4 with different size two lists}. First, we show that~$(G, \Psi, k)$ is a yes-instance if and only if~$(G', \Psi', k')$ is a yes-instance.  

\textbf{Case 1:~$\Psi(e_1) = \Psi(e_2)$.} Since Rule~\ref{Rule: split disjoint neighbors} is exhaustively applied, we can assume without loss of generality that~$\Psi(e_1) = \Psi(e_2) = \meng{1, 2}, \Psi(e_3) = \meng{2,3}$ and thus~$\Psi'(e_1) = \Psi'(e_2) = \meng{2}, \Psi'(e_3) = \meng{2,3}$. 

$(\Rightarrow)$ Let~$L$ be a proper~$\Psi$-satisfying labeling with at most~$k$ weak edges for~$G$. Initialize $L'$ with $L'|_{E \setminus E(P)} = L|_{E \setminus E(P)}$.
First, assume that~$\Q^P_L$ contains at least one weak color. We define~$L'(e_1) = L'(e_3) = 0$ and~$L'(e_2) = 2$. Obviously, $L'$ is a proper~$\Psi'$-satisfying labeling for~$G'$ with at most~$k-1+2=k'$ weak edges.
Second, assume that~$\Q^P_L$ contains no weak color. Then,~$\Q^P_L \in \meng{(1,2,3)$,$ (2,1,2)$,$(2,1,3)}$. 
If~$\Q^P_L = (1,2,3)$, define~$L'(e_1) = 0,L'(e_2) = 2$ and~$L'(e_3) = 3$. Otherwise, define~$L'(e_1) = 2,L'(e_2) = 0$ and~$L'(e_3) = L(e_3)$. In both cases it is obvious that $L'$ is a proper~$\Psi'$-satisfying labeling for~$G'$ with at most~$k+1=k'$ weak edges.

$(\Leftarrow)$ Let~$L'$ be a proper~$\Psi'$-satisfying labeling with at most~$k'$ weak edges for~$G'$. Since~$\Psi'(e_1) = \Psi'(e_2) = \meng{2}$ it follows that~$Q^P_{L'}$ contains at least one weak color. Initialize $L$ with $L|_{E \setminus E(P)} = L'|_{E \setminus E(P)}$.
First, assume that~$Q^P_{L'}$ contains at least two weak colors. We choose two arbitrary colors~$c_x \in \Psi(e_3) \setminus \out^{E \setminus E(P)}_{L'}(e_3),c_y \in \Psi(e_2) \setminus \meng{c_x}$ and define~$L(e_1) = 0, L(e_2) = c_y, L(e_3) = c_x$. Obviously, $L$ is a proper~$\Psi$-satisfying labeling for~$G$ with at most~$k'-1=k$ weak edges.
Second, assume that~$Q^P_{L'}$ contains exactly one weak color. Then,~$Q^P_{L'} \in \meng{(2,0,2)$,$(2,0,3)$,$(0,2,3)}$.
If~$\Q^P_{L'} = (2,0,2)$, define~$L(e_1) = L(e_3) = 2$ and~$L(e_2) = 1$. Otherwise, choose an arbitrary color~$c_x \in \Psi(e_1) \setminus \out^{E \setminus E(P)}_{L'}(e_1)$, let~$c_y$ be the unique remaining color in~$\Psi(e_2) \setminus \meng{c_x}$ and define~$L(e_1) = c_x,L(e_2) = c_y$ and~$L(e_3) = 3$. In both cases it is obvious that $L$ is a proper~$\Psi'$-satisfying labeling for~$G'$ with at most~$k'-1=k$ weak edges.

\textbf{Case 2:~$\Psi(e_1) \neq \Psi(e_2), \Psi(e_1) \cap \Psi(e_3) = \emptyset$.} Since Rule~\ref{Rule: split disjoint neighbors} is exhaustively applied, we can assume without loss of generality that~$\Psi(e_1) = \meng{1,2},  \Psi(e_2) = \meng{2, 3}, \Psi(e_3) = \meng{3,4}$ and thus~$\Psi'(e_1) = \Psi'(e_2) = \meng{1}, \Psi'(e_3) = \meng{1,4}$. 

$(\Rightarrow)$ Let~$L$ be a proper~$\Psi$-satisfying labeling with at most~$k$ weak edges for~$G$. 
Initialize $L'$ with $L'|_{E \setminus E(P)} = L|_{E \setminus E(P)}$.
First, assume~$\Q^P_L$ contains at least one weak color. We define~$L'(e_1) = L'(e_3) = 0$ and~$L'(e_2) = 1$. Obviously, $L'$ is a proper~$\Psi'$-satisfying labeling for~$G'$ with at most~$k-1+2=k'$ weak edges.
Second, assume that~$\Q^P_L$ contains no weak color. Then,~$\Q^P_L \in \meng{(2,3,4)$,$(1,2,3)$,$(1,2,4)$,$(1,3,4)}$.
If~$\Q^P_L = (2,3,4)$, define~$L'(e_1) = 0,L'(e_2) = 1$ and~$L'(e_3) = 4$. Otherwise, choose an arbitrary color~$c_x \in \Psi(e_3) \setminus \out^{E \setminus E(P)}_{L}(e_3)$ and define~$L'(e_1) = 1,L'(e_2) = 0$ and~$L'(e_3) = c_x$. In both cases it is obvious that $L'$ is a proper~$\Psi'$-satisfying labeling for~$G'$ with at most~$k+1=k'$ weak edges.

$(\Leftarrow)$ Let~$L'$ be a proper~$\Psi'$-satisfying labeling with at most~$k'$ weak edges for~$G'$. Since~$\Psi'(e_1) = \Psi'(e_2) = \meng{1}$ it follows that~$Q^P_{L'}$ contains at least one weak color.
Initialize $L$ with $L|_{E \setminus E(P)} = L'|_{E \setminus E(P)}$.
First, assume~$Q^P_{L'}$ contains at least two weak colors. We choose two arbitrary colors~$c_x \in \Psi(e_3) \setminus \out^{E \setminus E(P)}_{L'}(e_3),c_y \in \Psi(e_2) \setminus \meng{c_x}$ and we define~$L(e_1) = 0, L(e_2) = c_y, L(e_3) = c_x$. Obviously, $L$ is a proper~$\Psi$-satisfying labeling for~$G$ with at most~$k'-1=k$ weak edges.
Second, assume that~$Q^P_{L'}$ contains exactly one weak color. Then,~$Q^P_{L'} \in \meng{(1,0,1),(1,0,4),(0,1,4)}$.
If~$\Q^P_{L'} = (0,1,4)$, then choose an arbitrary color~$c_x \in \Psi(e_1) \setminus \out^{E \setminus E(P)}_{L'}(e_1)$ and define~$L(e_1) = c_x,L(e_2) = 3$ and~$L(e_3) = 4$. Otherwise, choose two arbitrary colors~$c_x \in \Psi(e_3) \setminus \out^{E \setminus E(P)}_{L'}(e_3), c_y \in \Psi(e_2) \setminus \meng{c_x}$ and define~$L(e_1) = 1,L(e_2) = c_y$ and~$L(e_3) = c_x$. In both cases it is obvious that $L$ is a proper~$\Psi'$-satisfying labeling for~$G'$ with at most~$k'-1=k$ weak edges.

\textbf{Case 3:~$\Psi(e_1) \neq \Psi(e_2), \Psi(e_1) \cap \Psi(e_3) \neq \emptyset$.} Since Rule~\ref{Rule: split disjoint neighbors} is exhaustively applied, we can assume without loss of generality that~$\Psi(e_1) = \meng{1,2},  \Psi(e_2) = \meng{2, 3}, \Psi(e_3) = \meng{3,1}$ and thus~$\Psi'(e_1) = \Psi'(e') = \meng{1}$. Furthermore, let~$P' := (v_1,v_2,v_4)$ 

$(\Rightarrow)$ Let~$L$ be a proper~$\Psi$-satisfying labeling with at most~$k$ weak edges for~$G$. 
Initialize $L$ with $L|_{E \setminus E(P)} = L'|_{E \setminus E(P)}$.
First, assume that~$\Q^P_L$ contains at least one weak color. We define~$L'(e_1) = L'(e') = 0$. Obviously, $L'$ is a proper~$\Psi'$-satisfying labeling for~$G'$ with at most~$k-1+2=k'$ weak edges.
Second, assume that~$\Q^P_L$ contains no weak color. Then,~$\Q^P_L \in \{(1,2,1)$,$ (1,2,3)$,$(1,3,1)$,$(2,3,1)\}$.
If~$\Q^P_L = (2,3,1)$, define~$L'(e_1) = 0$ and~$L'(e') = 1$. Otherwise, define~$L'(e_1) = 1$ and~$L'(e') = 0$. In both cases it is obvious that $L'$ is a proper~$\Psi'$-satisfying labeling for~$G'$ with at most~$k+1=k'$ weak edges.

$(\Leftarrow)$ Let~$L'$ be a proper~$\Psi'$-satisfying labeling with at most~$k'$ weak edges for~$G'$. Since~$\Psi'(e_1) = \Psi'(e') = \meng{1}$ it follows that~$Q^{P'}_{L'}$ contains at least one weak color and so~$Q^{P'}_{L'} \in \meng{(0,0), (1,0), (0,1)}$.
Initialize $L'$ with $L'|_{E \setminus E(P)} = L|_{E \setminus E(P)}$.
First, assume that~$Q^{P'}_{L'} = (0,0)$. We choose an arbitrary color~$c_x \in \Psi(e_3) \setminus \out^{E \setminus E(P)}_{L'}(e_3)$ and we define~$L(e_1) = 0, L(e_2) = 2$ and $ L(e_3) = c_x$. Since~$2\not\in \Psi(e_3)$, it is obvious that $L$ is a proper~$\Psi$-satisfying labeling for~$G$ with at most~$k-2+1=k$ weak edges.
Second, assume that~$Q^{P'}_{L'} = (1,0)$. We choose an arbitrary color~$c_x \in \Psi(e_3) \setminus \out^{E \setminus E(P)}_{L'}(e_3)$ and we define~$L(e_1) = 1, L(e_2) = 2$ and $L(e_3) = c_x$. Since~$2\not\in \Psi(e_3)$, it is obvious that $L$ is a proper~$\Psi$-satisfying labeling for~$G$ with at most~$k'-1=k$ weak edges.
Finally, assume that~$Q^{P'}_{L'} = (0,1)$. We choose an arbitrary color~$c_x \in \Psi(e_1) \setminus \out^{E \setminus E(P)}_{L'}(e_1)$ and we define~$L(e_1) = c_x, L(e_2) = 3$ and $L(e_3) = 1$. Since~$2\not\in \Psi(e_1)$, it is obvious that $L$ is a proper~$\Psi$-satisfying labeling for~$G$ with at most~$k'-1=k$ weak edges.

Next, we show the running time. Every application of Rule~\ref{Rule: P4 with different size two lists} decreases the number of edges that have a set of exactly two allowed colors by at least two and a BDSP~$P$ that fulfills the conditions of Rule~\ref{Rule: P4 with different size two lists} can be found in~$\Oh(n)$~time. Consequently, Rule~\ref{Rule: P4 with different size two lists} can be exhaustively applied in~$\Oh(n^2)$~time.
\end{proof}

From this point onwards we assume that Rules~\ref{Rule: delete full or empty edges}--\ref{Rule: P4 with different size two lists} are exhaustively applied. Since every application of Rule~\ref{Rule: P4 with different size two lists} decreases the number of edges that have two allowed colors by at least two, we get that every BDP~$P$ of length at least five contains at most two edges that have two allowed colors. Together with Rule~\ref{Rule: same lists max length two} we get that~$P$ has at most two edges that have only one allowed color. This gives us that every BDP has length at most five. By this fact it is possible to show that Rules~\ref{Rule: delete full or empty edges} --\ref{Rule: P4 with different size two lists} already give a kernel with at most~$13\pe_2$ edges. But since the linear factor can be improved to~$11\pe_2$, we first present a  rule to reduce the length of BDP to at most four.

\begin{redrule}\label{Rule: size one-two-two}
If~$(G, \Psi, k)$ contains a proper BDSP~$P = (v_1, v_2, v_3, v_4)$ with~$|\Psi(e_1)| = 1, |\Psi(e_2)| = |\Psi(e_3)| = 2$ where~$e_i := \meng{v_i, v_{i+1}}$ for all~$i \in \meng{1,2,3}$, do the following 
\begin{itemize}
\item\label{Case: same lists}\textbf{If~$ \Psi(e_2) = \Psi(e_3)$,} then remove~$v_2, v_3$, add a new edge~$e' := \meng{v_1, v_4}$ and set~$\Psi(e') := \Psi(e_1)$.
\item\label{Case: diff lists}\textbf{If~$ \Psi(e_2) \neq \Psi(e_3)$.} then remove~$v_3$ from~$G$, add a new edge~$e' := \meng{v_2, v_4}$ and set~$\Psi(e') := (\Psi(e_2) \cup \Psi(e_3)) \setminus (\Psi(e_2) \cap \Psi(e_3))$.
\end{itemize}
\end{redrule}

\begin{proposition}
Rule~\ref{Rule: size one-two-two} is safe and can be exhaustively applied in~$\Oh(n^2)$~time.
\end{proposition}
\begin{proof}
Let~$(G, \Psi, k)$ be an \textsc{EL-ECS} instance with~$P = (v_1, v_2, v_3, v_4)$ satisfying the conditions of Rule~\ref{Rule: size one-two-two}. Furthermore, let~$(G', \Psi', k)$ be the modified instance of \textsc{EL-ECS} constructed by Rule~\ref{Rule: size one-two-two}. First, we show that~$(G, \Psi, k)$ is a yes-instance if and only if~$(G', \Psi', k)$ is a yes-instance.  

\textbf{Case 1:~$ \Psi(e_2) = \Psi(e_3)$.} Since Rule~\ref{Rule: split disjoint neighbors} is exhaustively applied we can assume without loss of generality that~$\Psi(e_1) = \meng{1}, \Psi(e_2) = \Psi(e_3) = \meng{1,2}$ and thus~$\Psi'(e') = \meng{1}$. We set~$P_2 := (v_1,v_4)$. Initialize $L'$ with $L'|_{E \setminus E(P)} = L|_{E \setminus E(P)}$.

$(\Rightarrow)$ Let~$L$ be a proper~$\Psi$-satisfying labeling with at most~$k$ weak edges for~$G$. 
First, assume that~$\Q^P_L$ contains no weak color. Since~$\Psi(e_1) = \meng{1}$ it follows that~$\Q^P_L = (1,2,1)$. We define~$L'(e') = 1$. Since~$1 \not \in \out^{E \setminus E(P)}_L(e_1) \cup \out^{E \setminus E(P)}_L(e_3)$ it is obvious that $L'$ is a proper~$\Psi'$-satisfying labeling for~$G'$ with at most~$k$ weak edges.
Second, assume that~$\Q^P_L$ contains at least one weak color. We define~$L'(e') = 0$. Obviously, $L'$ is a proper~$\Psi'$-satisfying labeling for~$G'$ with at most~$k-1+1=k$ weak edges.

$(\Leftarrow)$ Let~$L'$ be a proper~$\Psi'$-satisfying labeling with at most~$k$ weak edges for~$G'$. Initialize $L$ with $L|_{E \setminus E(P)} = L'|_{E \setminus E(P)}$.
First, assume that~$L'(e') = 0$. We choose an arbitrary color~$c_x \in \Psi(e_3) \setminus \out^{E \setminus E(P)}_{L'}(e_3)$, set~$c_y$ as the unique remaining color of~$\Psi(e_2) \setminus \meng{c_x}$ and we define~$L(e_1) = 0, L(e_2) = c_y$ and $L(e_3) = c_x$. Obviously, $L$ is a proper~$\Psi$-satisfying labeling for~$G$ with at most~$k-1+1=k$ weak edges.
Second, assume that~$L'(e') = 1$. Since~$1 \not \in \out^{E \setminus E(P)}_{L'}(e_1) \cup \out^{E \setminus E(P)}_{L'}(e_3)$ we define~$L(e_1) = L(e_3) = 1$ and $L(e_2) = 2$. Obviously, $L$ is a proper~$\Psi$-satisfying labeling for~$G$ with at most~$k$ weak edges.

\textbf{Case 2:~$ \Psi(e_2) \neq \Psi(e_3)$.} Since Rule~\ref{Rule: split disjoint neighbors} is exhaustively applied we can assume without loss of generality that~$\Psi(e_1) = \meng{1}, \Psi(e_2) = \meng{1,2}, \Psi(e_3) = \meng{2,3}$ and thus~$\Psi'(e') = \meng{1,3}$. We let~$P_2 := (v_1, v_2 ,v_4)$.

$(\Rightarrow)$ Let~$L$ be a proper~$\Psi$-satisfying labeling with at most~$k$ weak edges for~$G$. 
Initialize $L'$ with $L'|_{E \setminus E(P)} = L|_{E \setminus E(P)}$.
First, assume that~$\Q^P_L$ contains no weak color. Since~$\Psi(e_1) = \meng{1}$ it follows that~$\Q^P_L = (1,2,3)$ so we define~$L'(e_1) = 1$ and~$L'(e') = 3$. Since~$1 \not \in \out^{E \setminus E(P)}_L(e_1), 3 \not \in \out^{E \setminus E(P)}_L(e_3)$ it is obvious that $L'$ is a proper~$\Psi'$-satisfying labeling for~$G'$ with at most~$k$ weak edges.
Second, assume that~$\Q^P_L$ contains at least one weak color. Choose an arbitrary color~$c_x \in \Psi'(e') \setminus \out^{E \setminus E(P)}_L(e_3)$ and define~$L'(e_1) = 0$ and~$L'(e') = c_x$. Obviously, $L'$ is a proper~$\Psi'$-satisfying labeling for~$G'$ with at most~$k-1+1=k$ weak edges.

$(\Leftarrow)$ Let~$L'$ be a proper~$\Psi'$-satisfying labeling with at most~$k$ weak edges for~$G'$. Initialize $L$ with $L|_{E \setminus E(P)} = L'|_{E \setminus E(P)}$.
First, assume that~$\Q^{P_2}_{L'}$ contains no weak color. Since~$\Psi(e_1) = \meng{1}$ it follows that~$\Q^P_L = (1,3)$ so we define~$L(e_1) = 1, L(e_2) = 2$ and~$L(e_3) = 3$. Since~$1 \not \in \out^{E \setminus E(P)}_{L'}(e_1), 3 \not \in \out^{E \setminus E(P)}_{L'}(e_3)$ it is obvious that $L$ is a proper~$\Psi$-satisfying labeling for~$G$ with at most~$k$ weak edges.
Second, assume that~$\Q^{P_2}_{L'}$ contains at least one weak color. Choose an arbitrary color~$c_x \in \Psi(e_3) \setminus \out^{E' \setminus E(P_2)}_{L'}(e')$ and define~$L(e_1) = 0, L(e_2) = 1$ and~$L(e_3) = c_x$. Obviously, $L$ is a proper~$\Psi$-satisfying labeling for~$G$ with at most~$k-1+1=k$ weak edges since~$c_x \not \in \out^{E' \setminus E(P_2)}_{L'}(e')$.

Next, we show the running time. Every application of Rule~\ref{Rule: size one-two-two} decreases the number of edges in~$G$ by one and a BDSP~$P$ that fulfills the conditions of Rule~\ref{Rule: size one-two-two} can be found in~$\Oh(n)$~time. Consequently, Rule~\ref{Rule: size one-two-two} can be exhaustively applied in~$\Oh(n^2)$~time.
\end{proof}

\begin{proposition}
  \label{prop:el-ecs-kernel}
For all~$c\in \mathds{N}$, \textsc{EL-ECS} and \textsc{EL-Multi-STC} admit an~$11\pe_2$-edge and~$10\pe_2$-vertex kernel for \textsc{EL-ECS} that can be computed in~$\Oh(n^2)$~time.
\end{proposition}
\begin{proof}
As argued in Section~\ref{sec:ed},~$|\core| \leq 2\pe_2$ and that there are at most~$5\pe_2$ edges incident to core vertices. Since Rules~\ref{Rule: delete full or empty edges} -~\ref{Rule: cycles} are exhaustively applied, every periphery component is either an isolated~$K_1$ or contains at least one close vertex. Hence, it is easy to see that every edge that is not incident with at least one core vertex has to lie on an BDP. The set of close vertices~$A^*$ has size at most~$4\pe_2$ since every core vertex has at most two neighbors that are not in the core, so there are at most~$4 \pe_2$ open BDPs and at most~$2 \pe_2$ non-open BDPs. By the facts that Rules~\ref{Rule: delete full or empty edges} -~\ref{Rule: size one at open end} reduced the size of open BDPs to at most two, Rules~\ref{Rule: delete full or empty edges} -~\ref{Rule: size one-two-two} reduced the size of non-open BDPs to at most four, so there are at most~$3 * 2 \pe_2 = 6 \pe_2$ edges that are not connected to core vertices. Altogether, we get that there are at most~$11\pe_2$ edges and at most~$10\pe_2$ vertices in the reduced instance of \textsc{EL-ECS} after Rules~\ref{Rule: delete full or empty edges} -~\ref{Rule: size one-two-two} are exhaustively applied. 
\end{proof}

Some of the previous reduction rules may look strange in a way that they have restriction under which they should not be applied but these restrictions were neither used to prove the correctness nor the running time of these reduction rules. The reason for this is, that these reduction rules should also work for \textsc{EL-Multi-STC}. To prove the same kernel we first give a reduction rule that solves \textsc{EL-Multi-STC} on all isolated triangles in~$\Oh(n)$~time and show afterwards, that if there is a non-isolated triangle in~$G$, it is also contained in~$G'$ and vice versa.

\begin{redrule}\label{Rule: K3 for stc}
If~$(G, \Psi, k)$ contains an isolated triangle consisting of the vertices~$v_0, v_1, v_2$, then remove all three vertices from~$G$ and decrease~$k$ by the number of edges~$e$ in~$G[\meng{v_1,v_2,v_3}]$ with~$\Psi(e)=\emptyset$.
\end{redrule}
\begin{proposition}
Rule~\ref{Rule: K3 for stc} is safe and can be exhaustively applied in~$\Oh(n)$~time.
\end{proposition}
 
The correctness of Rule~\ref{Rule: K3 for stc} follows directly from the fact, that every labeling on an isolated triangle is a proper STC labeling. So we can assume from now on that Rule~\ref{Rule: K3 for stc} is exhaustively applied and thus~$G$ does not contain any isolated triangle.

\begin{lemma}\label{Lemma: K3 safe}
Let~$P=(v_1, \dots, v_r), r \geq 2$ be a BDP in~$G$. Then there is no~$v \in V(P)$ such that~$v$ forms a triangle with two other vertices of~$G$ unless~$P$ is a non-open BDP and~$r = 2$.
\end{lemma}
\begin{proof}
Proof by contradiction. We assume that there is a vertex~$v\in V(P)$ that forms a triangle with two other vertices in~$G$ and~$P$ is an open BDP or~$r = 2$.

\textbf{Case 1:}~$P$ is an open BDP\textbf{.} We can assume without loss of generality that~$\deg(v_1) < 2$. Thus it is obvious that~$v_1$ can not form a triangle in~$G$. Since~$v_1 \in N(v_2)$ and~$|N(v_2)| \leq 2$, neither can~$v_2$. So we can prove by  induction that for all~$v_x \in V(P)$ it holds that~$v_x$ can not form a triangle in~$G$. This contradicts the assumption that there is such a vertex~$v$ in an open BDP.

\textbf{Case 2:} $r = 3$\textbf{.} Since~$v_2 \in N(v_1) \cap N(v_3)$ it remains to show that~$v_1, v_2, v_3$ do not form a triangle in~$G$ which is equivalent to~$v_3 \not \in N(v_1)$. So we assume that~$v_3 \in N(v_1)$. By the fact that~$|N(v_i)| \leq 2$ for all~$i \in \meng{1,2,3}$, it follows that~$v_1, v_2, v_3$ form an isolated triangle in~$G$ which contradicts the fact that~$P$ is a BDP.

\textbf{Case 3:}~$r > 3$\textbf{.} Since~$v_2 \in N(v_1) \cap N(v_3)$ and~$v_3 \not \in N(v_1)$ it is obvious that~$v_2$ can not form a triangle in~$G$ with its neighbors and neither can~$v_1$ nor~$v_3$. By induction no~$v_x \in V(P)$ can form a triangle with its neighbors and thus it cannot form a triangle in~$G$. This contradicts the assumption that there is such a vertex~$v$.
\end{proof}

\begin{proposition}\label{Prop: safe for stc}
The Rules~\ref{Rule: delete full or empty edges} -~\ref{Rule: size one-two-two} are safe for \textsc{EL-Multi-STC} if the instance is already reduced with respect to Rule~\ref{Rule: K3 for stc}.
\end{proposition}
\begin{proof}
Let~$(G, \Psi, k)$ be an \textsc{EL-Multi-STC} instance reduced with respect to Rule~\ref{Rule: K3 for stc}. With Lemma~\ref{Lemma: K3 safe} we know that for every edge that lies on a BDP~$P$ and on a triangle at the same time it holds that~$P$ is a non-open BDP of length exactly two. So we will show in the following, that none of the Rules~\ref{Rule: delete full or empty edges} -~\ref{Rule: size one-two-two} modifies a non-open BDP of length two or decreases a non-open BDP to one of length of two. For Rule~\ref{Rule: delete full or empty edges} this is obvious since this rule can not be applied on edges that lie on non-open BDPs of length two. Since the Rules~\ref{Rule: split disjoint neighbors} -~\ref{Rule: size one at open end} and~\ref{Rule: same lists max length two} -~\ref{Rule: size one-two-two} can not be applied on BDPs of length two and only decrease a BDSP to a length of two if it is a proper BDSP, all these rules are safe with respect to Lemma~\ref{Lemma: K3 safe}. The only thing left to show is that Rule~\ref{Rule: cycles} is also safe for \textsc{EL-Multi-STC}. By the fact that Rule~\ref{Rule: K3 for stc} is exhaustively applied, there are no isolated triangle in~$G$ and thus, for isolated cycles~$C$ it holds that~$G[C]$ is triangle-free. Since \textsc{EL-Multi-STC} is equivalent to \textsc{EL-ECS} on triangle-free graphs, it follows that Rule~\ref{Rule: cycles} is also safe for \textsc{EL-Multi-STC}.
\end{proof}

With this proposition, it is clear the previous reduction rules also admit the same linear edge kernel for \textsc{EL-Multi-STC}. Thus, Propositions~\ref{prop:el-ecs-kernel} and~\ref{Prop: safe for stc} give our main result for this section.

\elkernel*

\fi




\begin{thebibliography}{10}

\bibitem{BGKS19}
Laurent Bulteau, Niels Gr{\"{u}}ttemeier, Christian Komusiewicz, and Manuel
  Sorge.
\newblock Your rugby mates don't need to know your colleagues: Triadic closure
  with edge colors.
\newblock In {\em Proc. 11th {CIAC}}, volume 11485 of {\em LNCS}, pages
  99--111. Springer, 2019.

\bibitem{CE91}
Leizhen Cai and John~A. Ellis.
\newblock {NP}-completeness of edge-colouring some restricted graphs.
\newblock {\em Discrete Appl. Math.}, 30(1):15--27, 1991.

\bibitem{CFJ04}
Benny Chor, Mike Fellows, and David~W. Juedes.
\newblock Linear kernels in linear time, or how to save $k$ colors in
  ${O}(n^2)$~steps.
\newblock In {\em Proc.~30th WG}, volume 3353 of {\em LNCS}, pages 257--269.
  Springer, 2004.

\bibitem{CH82}
Richard Cole and John~E. Hopcroft.
\newblock On edge coloring bipartite graphs.
\newblock {\em {SIAM} J. Comput.}, 11(3):540--546, 1982.

\bibitem{FKLMPPS15}
Marek Cygan, Fedor~V. Fomin, Lukasz Kowalik, Daniel Lokshtanov, D{\'{a}}niel
  Marx, Marcin Pilipczuk, Michal Pilipczuk, and Saket Saurabh.
\newblock {\em Parameterized Algorithms}.
\newblock Springer, 2015.

\bibitem{FOW02}
Uriel Feige, Eran Ofek, and Udi Wieder.
\newblock Approximating maximum edge coloring in multigraphs.
\newblock In {\em Proc.~5th~APPROX}, volume 2462 of {\em LNCS}, pages 108--121.
  Springer, 2002.

\bibitem{Gab83}
Harold~N. Gabow.
\newblock An efficient reduction technique for degree-constrained subgraph and
  bidirected network flow problems.
\newblock In {\em Proc. 15th~STOC}, pages 448--456. {ACM}, 1983.

\bibitem{GHK+18}
Petr~A. Golovach, Pinar Heggernes, Athanasios~L. Konstantinidis, Paloma~T.
  Lima, and Charis Papadopoulos.
\newblock Parameterized aspects of strong subgraph closure.
\newblock In {\em Proc.~16th~SWAT}, volume 101 of {\em LIPIcs}, pages
  23:1--23:13. Schloss Dagstuhl - Leibniz-Zentrum fuer Informatik, 2018.

\bibitem{GK18}
Niels Gr{\"{u}}ttemeier and Christian Komusiewicz.
\newblock On the relation of strong triadic closure and cluster deletion.
\newblock In {\em Proc.~44th~WG}, volume 11159 of {\em LNCS}, pages 239--251.
  Springer, 2018.

\bibitem{GHN04}
Jiong Guo, Falk H{\"{u}}ffner, and Rolf Niedermeier.
\newblock A structural view on parameterizing problems: Distance from
  triviality.
\newblock In {\em Proc.~1st IWPEC}, volume 3162 of {\em LNCS}, pages 162--173.
  Springer, 2004.

\bibitem{HK86}
S.~Louis Hakimi and Oded Kariv.
\newblock A generalization of edge-coloring in graphs.
\newblock {\em J. Graph Theor.}, 10(2):139--154, 1986.

\bibitem{Hol81}
Ian Holyer.
\newblock The {{NP}}-{{Completeness}} of {{Edge}}-{{Coloring}}.
\newblock {\em {SIAM} J. Comput.}, 10(4):718--720, 1981.

\bibitem{JT11}
Tommy~R Jensen and Bjarne Toft.
\newblock {\em Graph coloring problems}, volume~39.
\newblock John Wiley \& Sons, 2011.

\bibitem{KK14}
Marcin~Jakub Kaminski and Lukasz Kowalik.
\newblock Beyond the {V}izing's bound for at most seven colors.
\newblock {\em {SIAM} J. Discrete Math.}, 28(3):1334--1362, 2014.

\bibitem{KNP17}
Athanasios~L. Konstantinidis, Stavros~D. Nikolopoulos, and Charis Papadopoulos.
\newblock Strong triadic closure in cographs and graphs of low maximum degree.
\newblock In {\em Proc.~23rd COCOON}, volume 10392 of {\em LNCS}, pages
  346--358. Springer, 2017.

\bibitem{KP17}
Athanasios~L. Konstantinidis and Charis Papadopoulos.
\newblock {Maximizing the Strong Triadic Closure in Split Graphs and Proper
  Interval Graphs}.
\newblock In {\em Proc.~28th~ISAAC}, volume~92 of {\em LIPIcs}, pages
  53:1--53:12. Schloss Dagstuhl--Leibniz-Zentrum fuer Informatik, 2017.

\bibitem{Kos09}
Adrian Kosowski.
\newblock Approximating the maximum 2- and 3-edge-colorable subgraph problems.
\newblock {\em Discrete Appl. Math.}, 157(17):3593--3600, 2009.

\bibitem{Kow09}
Lukasz Kowalik.
\newblock Improved edge-coloring with three colors.
\newblock {\em Theor. Comput. Sci.}, 410(38-40):3733--3742, 2009.

\bibitem{KL16}
Mithilesh Kumar and Daniel Lokshtanov.
\newblock A $2\ell k$ kernel for $\ell$-component order connectivity.
\newblock In {\em Proc.~11th {IPEC}}, pages 20:1--20:14, 2016.

\bibitem{LG83}
Daniel Leven and Zvi Galil.
\newblock {NP} completeness of finding the chromatic index of regular graphs.
\newblock {\em J. Algorithms}, 4(1):35--44, 1983.

\bibitem{Riz09}
Romeo Rizzi.
\newblock Approximating the maximum 3-edge-colorable subgraph problem.
\newblock {\em Discrete Math.}, 309(12):4166--4170, 2009.

\bibitem{R05}
Elena~Prieto Rodr{\'\i}guez.
\newblock {\em Systematic kernelization in FPT algorithm design}.
\newblock PhD thesis, The University of Newcastle, 2005.

\bibitem{ST14}
Stavros Sintos and Panayiotis Tsaparas.
\newblock Using strong triadic closure to characterize ties in social networks.
\newblock In {\em Proc.~20th~KDD}, pages 1466--1475. ACM, 2014.

\bibitem{S12}
Michael Stiebitz, Diego Scheide, Bjarne Toft, and Lene~M Favrholdt.
\newblock {\em Graph edge coloring: Vizing's theorem and Goldberg's
  conjecture}, volume~75.
\newblock John Wiley \& Sons, 2012.

\bibitem{V64}
Vadim~G Vizing.
\newblock On an estimate of the chromatic class of a $p$-graph.
\newblock {\em Discret Analiz}, 3:25--30, 1964.

\end{thebibliography}

\end{document}